\documentclass[a4paper,11pt]{article}

\usepackage[utf8]{inputenc}

\usepackage{amsthm}
\usepackage{amsmath}
\usepackage{amssymb}
\usepackage{amsfonts}
\usepackage{mathrsfs}
\usepackage{verbatim}
\usepackage{xspace}

\newtheorem{theorem}{Theorem}[section]
\newtheorem{lemma}[theorem]{Lemma}

\newtheorem{definition}[theorem]{Definition}

\newtheorem{claim}[theorem]{Claim}

\newtheorem{obsv}[theorem]{Observation}

\newcommand{\Pm}[1]{\ensuremath{{\normalfont\textsf{P}}_{#1}}}
\newcommand{\Pmd}[2]{\ensuremath{{\normalfont\textsf{P}}_{#1,#2}}}

\newcommand{\cclass}[1]{{\normalfont\textsf{#1}}\xspace}
\newcommand{\cproblem}[1]{{\normalfont\textsc{#1}}\xspace}
\newcommand{\np}{\cclass{NP}}

\newcommand{\E}{\mathop{\mathbf{E}}}

\newcommand{\Opt}{\mathop{\mathrm{Opt}}\nolimits}

\newcommand{\F}[1]{\mathbb{F}_{#1}}

\newcommand{\N}{{\mathbb{N}}}
\newcommand{\R}{{\mathbb{R}}}

\newcommand{\bool}{ \left\{ 0,1 \right\} }
\newcommand{\boolean}{ \left\{ -1,1 \right\} }

\newcommand{\poly}{\mathop{\mathrm{poly}}}
\newcommand{\polylog}{\mathop{\mathrm{polylog}}}
\newcommand{\rank}{\mathop{\mathrm{rank}}}

\newcommand{\lindim}{\mathop{\mathrm{dim}}}
\newcommand{\codim}{\mathop{\mathrm{codim}}}

\newcommand{\LC}{\cproblem{Label-Cover}}

\newcommand{\LGCD}{\cproblem{Long-Code}}
\newcommand{\LDLC}{\cproblem{Low-Degree-Long-Code}}
\newcommand{\QDCD}{\cproblem{Quadratic-Code}}

\newcommand{\tsa}{\textsc{TSA}\xspace}

\newcommand{\nae}[1]{\textsc{NotAllEqual}$_#1$}

\newcommand{\ksat}[1]{\textsc{E}$#1$-\textsc{Sat}}
\newcommand{\csp}{\textsc{CSP}\xspace}

\newcommand{\weight}{{\mathrm{wt}}}
\newcommand{\supp}{{\mathrm{supp}}}

\begin{document}
\title{$2^{(\log N)^{1/10-o(1)}}$ Hardness for Hypergraph Coloring}
\author{Sangxia Huang
\thanks{\texttt{huang.sangxia@gmail.com} ~
EPFL. Research partially supported by ERC Advanced Investigator Grant 226203 and Swedish
Research Council.}}
\maketitle

\begin{abstract}
  We show that it is quasi-\np-hard to color $2$-colorable $8$-uniform
  hypergraphs with $2^{(\log N)^{1/10-o(1)}}$ colors, where $N$ is the number
  of vertices. There has been much focus on hardness of hypergraph coloring
  recently. In \cite{shortcodecol}, Guruswami, H{\aa}stad, Harsha, Srinivasan
  and Varma showed that it is quasi-\np-hard to color $2$-colorable $8$-uniform
  hypergraphs with $2^{2^{\Omega(\sqrt{\log\log N})}}$ colors.
  Their result is obtained by composing standard \LC with an inner-verifier
  based on \LDLC, using 
  Reed-Muller code testing results by Dinur and Guruswami \cite{dinurguruswamishortcode}.
  Using a different approach in \cite{kssp}, Khot and Saket constructed a new
  variant of \LC, and composed it with \QDCD to show quasi-\np-hardness
  of coloring $2$-colorable $12$-uniform hypergraphs with $2^{(\log N)^c}$ colors,
  for some $c$ around $1/20$. Their construction of \LC is based on a new notion
  of \emph{superposition complexity} for \csp instances.
  The composition with inner-verifier was subsequently improved by Varma, giving the same hardness result
  for $8$-uniform hypergraphs \cite{varmasp}.

  Our construction uses both \QDCD and \LDLC,
  and builds upon the work by Khot and Saket. We present a different approach
  to construct \csp instances with superposition hardness by observing that when the number
  of assignments is odd, satisfying a constraint in superposition is the same as \emph{odd-covering}
  a constraint.
  We employ \LDLC in order to keep the construction
  efficient. In the analysis, we also adapt and generalize one of the key theorems
  by Dinur and Guruswami \cite{dinurguruswamishortcode} in the context of
  analyzing probabilistically checkable proof systems.
\end{abstract}

\section{Introduction}\label{sec:intro}
For an integer $k \ge 2$, a $k$-uniform
hypergraph $H=(V,F)$ consists of vertex set $V$ and edge set
$F \subseteq \binom{V}{k}$. 
A set of vertices $S \subseteq V$ is an \emph{independent set} if
for all $f \in F$, $f \not\subseteq S$, i.e., no edge is completely
inside $S$.
A hypergraph is $q$-colorable if its vertices can be partitioned
into $q$ disjoint independent sets.

%Coloring a graph or a hypergraph using few colors is a classical combinatorial
%optimization problem, and is one of the most well-studied problems in theoretical
%computer science. It is also closely related
%to other problems such as finding maximum independent sets,
%PCPs with certain special properties, and also inapproximability of constraint
%satisfaction problems.
%In addition to being an important theoretical challenge, 
%graph coloring also has a number of applications such as scheduling and register allocation.

We use $\alpha(H)$ to denote the fractional size
of the maximum cardinality independent set of $H$,
also known as the \emph{fractional independence number},
and we use $\chi(H)$ to denote the minimum $q$ such that
$H$ is $q$-colorable. It is easy to verify that we have $\chi(H)\alpha(H) \ge 1$
for any $H$.

In the ordinary graph case, corresponding to $k=2$, 
deciding whether a graph $G$ has a $2$-coloring is the same
as deciding whether it is a bipartite graph, and can be easily solved in polynomial time. 
In general, however,
determining the chromatic number of a graph exactly is \np-hard \cite{gareyjohnson}.
In fact, even coloring 3-colorable graphs with 4 colors is \np-hard.
For general $q$-colorable graphs,
it is \np-hard to color with $q+2\lfloor \frac{q}{3} \rfloor-1$ colors
\cite{hard3,hard5}.
For sufficiently large $q$, it was shown that it is \np-hard to
color a $q$-colorable graph with $2^{\Omega(q^{1/3})}$ colors \cite{cubecolor},
improving on an earlier lower-bound of $q^{\frac{1}{25} \log q}$ 
by Khot \cite{khotcoloring}.
Assuming a variant of Khot's 2-to-1 Conjecture, Dinur, Mossel and Regev \cite{dmrcoloring} proved
that it is \np-hard to $q'$-color a $q$-colorable graph for any $3 \le q < q'$.
The dependency between the hardness of graph coloring and the parameters of 2-to-1 \LC
was made explicit and improved by Dinur and Shinkar \cite{improved2to1coloring},
who showed that it is \np-hard to $(\log n)^c$-color a $4$-colorable graph
for some constant $c>0$ assuming the 2-to-1 Conjecture.
As for algorithms, there have been many results as well
\cite{avicoloring,bergerrompelcoloring,kargermotwanisudan,blumkarger}.
For 3-colorable graphs, the best algorithm is by Kawarabayashi and Thorup \cite{thorup3col}
which uses $O(n^{0.19996})$ colors, based on results by Arora and Chlamtac \cite{arorachlamtaccoloring},
Chlamtac \cite{hierarchycoloring} and the earlier work of Kawarabayashi and Thorup \cite{focs12coloring}. 
As we can see, there is still a huge gap between the best approximation guarantee
and the best hardness result.

For $k \ge 3$, even determining
whether a $k$-uniform hypergraph has a $2$-coloring is \np-hard.
In terms of approximation algorithms, the best algorithm 
for $2$-colorable $3$-uniform hypergraphs still requires $n^{\Omega(1)}$ colors
\cite{knscoloring,akmhcoloring,cfcoloring}. 

From the hardness side,
the first super-constant hardness result was proved in
\cite{guruswami4uniformhyper}. The main result there is that for $2$-colorable $4$-uniform
hypergraphs, finding a coloring with any constant number
of colors is \np-hard, and finding a coloring with $O(\log\log n/\log\log\log n)$
colors is quasi-\np-hard.
For $2$-colorable $3$-uniform hypergraphs, a similar hardness result was proved in \cite{drs05}.
Khot \cite{khotsmooth} proved that coloring $3$-colorable $3$-uniform
hypergraphs with any constant number of colors is hard, and 
for $q$-colorable $4$-uniform hypergraphs, coloring with $(\log n)^{\Omega(q)}$
colors is quasi-\np-hard for $q \ge 7$.
The analysis in \cite{guruswami4uniformhyper} was improved by Holmerin,
who proved that even finding an independent set of fractional size
$\Omega(\log\log\log n/\log\log n)$ is quasi-\np-hard \cite{hol4uni}.
The construction was further improved recently by Saket \cite{saketpolylog},
who proved that it is
quasi-\np-hard to find independent set of size $n/(\log n)^{\Omega(1)}$
in $2$-colorable $4$-uniform hypergraphs \cite{saketpolylog}.
There has also been work on the hardness of finding independent sets
in almost $2$-colorable hypergraphs --- hypergraphs that becomes $2$-colorable
after removing a small fraction of vertices. Much stronger result
is known, albeit at the cost of imperfect completeness. We refer
to \cite{kshyperalmost} for more details.

Recently, in \cite{shortcodecol}, Guruswami, Harsha, H{\aa}stad, Srinivasan
and Varma proved the first super-polylogarithmic hardness result for hypergraph coloring,
showing hardness for coloring $2$-colorable $8$-uniform
hypergraphs with $2^{2^{\Omega(\sqrt{\log\log n})}}$ colors.
Their reduction uses the \LDLC proposed in \cite{shortcode},
based on techniques for testing Reed-Muller codes developed in \cite{dinurguruswamishortcode}.

Using a very different approach, Khot and Saket gave another exponential improvement
in \cite{kssp}, showing a quasi-\np-hardness for coloring $2$-colorable $12$-uniform
hypergraphs with $\exp( (\log n)^{\Omega(1)})$ colors, where the constant in $\Omega(1)$
is around $1/20$, although it might be improved with a more careful analysis of their reduction.
Part of their analysis was subsequently
simplified by Varma in \cite{varmasp} using ideas from \cite{shortcodecol}. 

In this work, we give another improvement for hardness of hypergraph coloring.
Our main result is as follows.
\begin{theorem}
  It is quasi-\np-hard to color a $2$-colorable
  $8$-uniform hypergraph of size $N$ with $2^{(\log N)^{1/10-o(1)}}$
  colors.
\end{theorem}

\subsection{Proof Overview}
We start by describing the PCP reduction of proving hypergraph
coloring hardness used in many previous works.
Most of these results show hardness of finding an independent
set of large fractional size. We can view the output of these 
reductions as \nae{k} \ \csp instances.
The variables correspond to the vertices
of a hypergraph, and the \nae{k} constraints correspond to the
hyperedges. Note that for hypergraph coloring results,
all variables appear positively in such instances and no negations are allowed.
An assignment that satisfies all the \nae{k}
constraints thus gives a perfect $2$-coloring for the hypergraph.
In the other direction, a set of vertices in the hypergraph
naturally corresponds to a $\bool$ assignment to the variables
in the \nae{k} instance, and the vertices form an independent
set if for all constraints in the \nae{k} instances, there is
at least $1$ variable that is assigned $0$.

The starting point of the reduction is usually some \LC 
hardness. We then encode the supposed labeling for
the \LC instance with some coding scheme, 
and design a PCP to test the consistency of the labeling.

One classical choice of encoding is the \LGCD, which
encodes $m$ bits of information with $2^{2^m}$ bits. 
This huge blowup makes it impossible to prove hardness results
better than $\polylog n$ via the \LC plus \LGCD approach.

A much more efficient encoding is the Hadamard code, which only uses
$2^m$ bits to encode $m$ bits of information. 
However, the disadvantage of the Hadamard code is that
one can only enforce linear constraints on the codewords, 
which means that we can only start from hard
problems involving only linear constraints,
and as a result, we lose perfect completeness and can only prove results
about almost coloring. 

The \LDLC proposed in \cite{shortcode} lies somewhere between
\LGCD and Hadamard code. We can view Hadamard code
as encoding $m$ bits by writing down the evaluation of all $m$-variable
functions of degree at most $1$ on these $m$ bits, and \LGCD as writing down the evaluation
of all possible $m$-variable functions --- that is, degree up to $m$ --- on these $m$ bits. 
\LDLC has a parameter
$d$, the degree, and the encoding writes down the evaluation of all polynomials
of degree at most $d$.
Dinur and Guruswami \cite{dinurguruswamishortcode}
obtained hardness result for a variant of hypergraph coloring based on \LDLC,
and the techniques were soon adapted in \cite{shortcodecol}
to get a hardness result of $2^{2^{\Omega(\sqrt{\log\log n})}}$.

The aforementioned result by Khot and Saket \cite{kssp}
uses \QDCD, which is the same as \LDLC with $d=2$. 
Their construction, however, is completely different from that
in \cite{shortcodecol}.

One can view the \QDCD used in \cite{kssp}
as the Hadamard encoding of matrix $M$ that is symmetric
and has rank $1$, that is, there exists some $u \in \F{2}^m$ such that
$M=u \otimes u$.
Khot and Saket described a $6$-query test such that if some encoding function
$f:\F{2}^{m \times m} \to \F{2}$ passes the test with non-trivial probability, 
then we can decode it into a low rank matrix. 

In order to use this encoding, it seems natural that one would like to construct
some variant of \LC where the labels are now matrices, with some
linear constraints on the entries of the matrices (since as discussed above
we are using Hadamard code to encode the matrices). In the completeness case,
we would like to have
some matrix labelings of rank $1$ that satisfies all linear constraints
on the vertices as well as projection constraints on the edges, and in the
soundness case, not even labelings with low rank matrices can satisfy more than
a small fraction of them.

Such \LC hardness result is not readily available.
Khot and Saket proposed the notion of \emph{superposition complexity} for quadratic
equations. Briefly speaking, let
$q(x)=c+\sum_{i=1}^{m} c_i x_i + \sum_{1 \le i < j \le m} c_{ij} x_i x_j=0$
be a quadratic equation on $m$ $\F{2}$-variables.
We say that $t$ assignments $a^{(1)},\cdots,a^{(t)} \in \F{2}^m$
satisfy the equation $q(x)=0$ \emph{in superposition} if
\[
c + \sum_{i=1}^{m} c_i \left( \sum_{l=1}^{t} a^{(l)}_i \right)
+ \sum_{1 \le i < j \le m} c_{ij} \left( \sum_{l=1}^{t} a^{(l)}_i a^{(l)}_j \right)
=0\,.
\]
If we have a system of quadratic equations, then we say that $t$ assignments
satisfy the system of quadratic equations in superposition if each
quadratic equation is satisfied in superposition.
Having a small number of assignments satisfying quadratic constraints in superposition
is exactly the same as having a symmetric low-rank matrix satisfying 
the linearized version of the constraints, as we discuss in more detail
in Section \ref{sec:oddcovprelim}.

Through a remarkable chain of reductions, Khot and Saket established
the inapproximability of quadratic equations with superposition complexity,
as well as the actual construction of the \LC with matrix labels.
They started with superposition hardness for \ksat{3} with gap of $1/n$, and use
low-degree testing and sum-check protocol like in the original proof
of the PCP theorem \cite{arorapcp1,arorapcp2}
to achieve a superposition hardness result for systems of quadratic equations
with good soundness and moderate blowup in size. This is then followed by a
\emph{Point versus Ruled Surface} test which produces the actual \LC instance.

The main contribution of this work is a much simpler and more efficient construction
of systems of quadratic equations with superposition and approximation hardness.
This is then coupled with a slight variant of the \emph{Point versus Ruled Surface}
test used by Khot and Saket to obtain hardness for \LC with matrix labels.

Let $t$ be some odd natural number. A set of $t$ assignments \emph{odd-covers} an equation
(or more generally, a constraint) if the number of assignments that satisfy
the equation is odd. We show in Section \ref{sec:oddcovprelim} that
the notion of odd-covering is equivalent to satisfaction
in superposition when the number of assignments is odd.
This viewpoint enables us to construct the kind of \LC instance
used in \cite{kssp} very easily.
In fact, the reduction in Section \ref{sec:octsa} looks very much like a classical
\csp inapproximability proof.

Although simpler, the above observation alone is not sufficient to give us
a hardness result better than \cite{kssp}. The issue here is that for the reduction
in Section \ref{sec:octsa} to work for our choice of parameters, the soundness of
the \LC that we start with needs to be sub-constant, and a typical \LGCD reduction
will again blow up the size of the instance by too much. Hence, for this step, we employ
\LDLC. Our technical contribution here is Theorem \ref{thm:lowdegreebias},
a generalization of the Reed-Muller code testing result of \cite{dinurguruswamishortcode}.

%\subsection{Subsequent Work}
%After the preliminary version of this work was made public on arXiv \cite{arxivversion},
%Khot, Saket and Thiruvenkatachari \cite{hyper6uni} improved our result by showing that
%the same hardness holds for $6$-uniform hypergraphs. The starting point of their construction
%is Theorem \ref{thm:octsa} of this work.
%The key improvement in their work is the construction of \LC instances with stronger properties than the ones we construct
%in Theorem \ref{thm:lcmatrix}.
%
%

\section{Preliminaries}\label{sec:oddcovprelim}
Before we discuss the relation between superposition, odd-covering
and low rank matrices, we define an operation on vectors and matrices
that we will use frequently.
\begin{definition}
  Define $D_1:\F{2}^{m+1} \to \F{2}^m$ as the operator that removes
  the first coordinate of a vector. Define $D_1$ similarly for matrices
  as the operator that removes the first row and column of a given matrix.
\end{definition}

\subsection{Superposition and Odd-Covering}
Khot and Saket \cite{kssp}
defined the notion of satisfying in superposition as follows.
\begin{definition}[Superposition]
  Let $a^{(1)},\cdots,a^{(t)} \in \F{2}^m$ be $t$ assignments and $q(x)=0$
  be a quadratic equation in $m$ $\F{2}$-variables with
  \[ q(x)=c+\sum_{i=1}^{m} c_i x_i + \sum_{1 \le i < j \le m} c_{ij} x_i x_j\,.  \]
  We say that the $t$ assignments satisfy the equation $q(x)=0$ 
  \emph{in superposition} if
  \[
  c + \sum_{i=1}^{m} c_i \left( \sum_{l=1}^{t} a^{(l)}_i \right)
  + \sum_{1 \le i < j \le m} c_{ij} \left( \sum_{l=1}^{t} a^{(l)}_i a^{(l)}_j \right)
  =0\,.
  \]
\end{definition}
\begin{definition}
  Given a system of quadratic equations $\{q_i(x)=0\}_{i=1}^{L}$ on variables
  $x_1,\cdots,x_m$, its \emph{superposition complexity} is the minimum number
  $t$, if it exists, such that there are $t$ assignments
  $a^{(1)},\cdots,a^{(t)} \in \F{2}^m$ that satisfy each equation $q_i(x)=0$
  in superposition.

  We define the \emph{odd superposition complexity} (or \emph{even superposition complexity})
  to be the minimum odd integer $t$ (or even integer $t$, respectively) such that
  there are $t$ assignments that satisfy all equations in superposition.
\end{definition}
Note that by simply adding all $0$ assignments, we can argue that the above three
notions of superposition complexity differ by at most $1$.

We now explain the relation between superposition complexity of quadratic
equations and low rank matrices. Assume for simplicity of exposition
that the quadratic equation $q(x)=0$ as defined above is homogeneous, 
that is, the constant term $c$ and the linear coefficients $c_i$ are all $0$. 

We can express a homogeneous quadratic equation $q(x)=0$ with a matrix by
defining $C \in \F{2}^{m \times m}$,
where $C_{ij}=c_{ij}$ for $1 \le i<j \le m$, and $C_{ij}=0$ otherwise.
Let $x=(x_1 \  x_2 \  \cdots \  x_m)$. Then $q(x)=0$ is the same as 
$\langle C, x \otimes x \rangle=x^T C x=0$, where $\langle \cdot,\cdot \rangle$
denotes the entry-wise dot product of two matrices.
Note that $x \otimes x$ is a symmetric rank-$1$ matrix. 

Suppose now that we have a symmetric matrix $A$ such that $\langle C,A \rangle=0$.
For a fixed $C$, this is a linear constraint on the entries of $A$.
If in addition $A$ has rank $1$, then there exists $x_a$, such that $A=x_a \otimes x_a$,
and by the above, we have that $x_a$ satisfies $q(x_a)=0$. Therefore, if
$A$ is a symmetric rank $1$ matrix and $\langle C,A \rangle=0$, then $A$
encodes an assignment that satisfies the quadratic equation $q(x)=0$.

The following decomposition lemma from \cite{kssp} illustrates the situation
when $A$ has low rank.
\begin{lemma}\label{lem:ksdecomp}
  Let $A \in \F{2}^{m \times m}$ be a symmetric matrix of rank $k$ over $\F{2}$.
  Then there exists $l \le 3k/2$ and vectors $v_1,\cdots,v_l$ in the column space of $A$,
  such that $A=\sum_{i=1}^{l} v_i \otimes v_i$.
\end{lemma}
Let $A$ be a low rank matrix such that $\langle C,A \rangle =0$ and
$v_1,\cdots,v_l$ be $l \le 3k/2$ assignments given by Lemma \ref{lem:ksdecomp}.
Then
\begin{align*}
  0=\langle C,A \rangle \ =&\  \sum_{t=1}^{l} \langle C,v_t \otimes v_t\rangle \\
  =&\ \sum_{t=1}^{l} \sum_{1 \le i < j \le m} c_{ij} v_{ti} v_{tj} \\
  =&\ \sum_{1 \le i < j \le m} c_{ij} \sum_{t=1}^{l} v_{ti} v_{tj} \,.
\end{align*}
Therefore we have that $v_1,\cdots,v_l$ satisfy $q(x)=0$ in superposition.

The notion we will now consider is the following, which we call 
\emph{odd-covering}.
\begin{definition}[Odd-covering]
  Let $a^{(1)},\cdots,a^{(t)} \in \F{2}^m$ be $t$ assignments and $q(x)=0$
  be a quadratic equation in $m$ $\F{2}$-variables as defined above.
  We say that the $t$ assignments \emph{odd-cover} the equation $q(x)=0$
  if the number of assignments $a^{(l)}$ that satisfies $q(a^{(l)})=0$
  is odd.
\end{definition}
The key observation is that odd-covering and satisfying in superposition
are equivalent when the number of assignments involved is odd.
\begin{lemma}
  Let $t$ be an odd integer and $a^{(1)},\cdots,a^{(t)} \in \F{2}^m$ be $t$ assignments,
  and $q(x)=0$ be a quadratic equation in $m$ $\F{2}$-variables as defined above.
  Then the $t$ assignments satisfy $q(x)=0$ in superposition
  if and only if the $t$ assignments odd-cover $q(x)=0$.
\end{lemma}
\begin{proof}
  Using the fact that $t$ is odd, we have the following
  \begin{align*}
    \sum_{l=1}^{t} q(a^{(l)}) \ =&\ 
    \sum_{l=1}^{t} \left( c + \sum_{i=1}^{m} c_i a^{(l)}_i 
    + \sum_{1 \le i<j \le m} c_{ij} a^{(l)}_i a^{(l)}_j \right) \\
    =&\  t \cdot c + \sum_{l=1}^{t} \sum_{i=1}^{m} c_i a^{(l)}_i
    + \sum_{l=1}^{t} \sum_{1 \le i<j \le m} c_{ij} a^{(l)}_i a^{(l)}_j \\
    =&\  c + \sum_{i=1}^{m} c_i \left( \sum_{l=1}^{t} a^{(l)}_i \right)
    + \sum_{1 \le i<j \le m} c_{ij} \left( \sum_{l=1}^{t} a^{(l)}_i a^{(l)}_j \right)\,.
  \end{align*}
  Now observe that the $t$ assignments odd-cover $q(x)=0$ if and only if
  the number of assignments that does not satisfy $q(x)=0$ is even,
  which is equivalent to saying that the left hand side of the above equation is $0$,
  and that by definition means that the $t$ assignments satisfy $q(x)=0$ in superposition.
\end{proof}

In the description above, we assumed that the quadratic equation $q(x)=0$ is
homogeneous, which allows us to encode it with a matrix $C \in \F{2}^{m \times m}$
and express the whole equation as $\langle C,A\rangle=0$, where $A=x \otimes x$.
For quadratic equations that are not homogeneous, we encode them with
a $(m+1) \times (m+1)$ matrix. In particular, for $q(x)=c+\sum c_i x_i+\sum c_{ij} x_i x_j=0$,
we have matrix $C$, where $C_{11}=c$, $C_{1i}=c_{i-1}$ for $i=2,\cdots,m+1$,
and $C_{ij}=c_{i-1,j-1}$ for $2 \le i < j \le m+1$.
As for the variable vector, we add to $x$ an entry that is always $1$.
\begin{definition}
  Given a matrix $A \in \F{2}^{(m+1) \times (m+1)}$. We say that $A$ is
  \emph{pseudo-quadratic} if the following holds:
  \begin{itemize}
    \item $A$ is symmetric.
    \item $A_{1,1}=1$.
    \item For all $i=2,\cdots,m+1$, $A_{1,i}=A_{i,1}=A_{i,i}$.
  \end{itemize}
\end{definition}
Note that for vector $v \in \F{2}^{m+1}$ such that $v_1=1$, 
$v \otimes v$ is a pseudo-quadratic rank-$1$ matrix.

We prove a stronger form of Lemma \ref{lem:ksdecomp} for pseudo-quadratic matrices
where we decode a low rank pseudo-quadratic matrix into an odd number of assignments.
\begin{lemma}\label{lem:odddecomp}
  Let $A \in \F{2}^{(m+1) \times (m+1)}$ be a pseudo-quadratic matrix of rank $k$
  over $\F{2}$. Then there exists an odd integer $k_0 < 3k/2+1$, and 
  vectors $v_1,\cdots,v_{k_0} \in \F{2}^{m+1}$, such that for all $i \in [k_0]$,
  $v_{i,1}=1$, and $A=\sum_{i=1}^{k_0} v_{i} \otimes v_{i}$. Moreover, for all $i \in [k_0]$,
  $D_1(v_i)$ is in the column space of $D_1(A)$.
\end{lemma}
\begin{proof}
  Let $A'=D_1(A)$.
  Note that $A'$ is symmetric and has rank at most $k$. Therefore by Lemma \ref{lem:ksdecomp},
  there exists $l<3k/2$ vectors $u_1,\cdots,u_l \in \F{2}^{m}$, such
  that $A'=\sum_{i=1}^{l} u_i \otimes u_i$.
  Now consider vectors $v_1,\cdots,v_l \in \F{2}^{m+1}$, where for each $i$,
  $v_{i,1}=1$ and $v_{i,j}=u_{i,j-1}$ for $j=2,\cdots,m+1$. 
  Let $A''=\sum_{i=1}^{l} v_{i} \otimes v_{i}$, and $B=A-A''$.
  For $j,j' \in \left\{ 2,\cdots,m+1 \right\}$, we have
  \[
  A''_{j,j'} = \sum_{i=1}^{l} v_{i,j} v_{i,j'} = \sum_{i=1}^{l} u_{i,j-1} u_{i,j'-1}
  =A'_{j-1,j'-1} = A_{j,j'}\,.
  \]
  Moreover, we have 
  \[
  A''_{1,j} = \sum_{i=1}^{l} v_{i,1} v_{i,j} = \sum_{i=1}^{l} v_{i,j} v_{i,j}
  =A''_{j,j} = A_{j,j}=A_{1,j}\,.
  \]
  We conclude that for all $(i,j) \ne (1,1)$, $A_{i,j}=A''_{i,j}$. 
  Note that $A''_{1,1}=(l \bmod 2)$. Therefore if $A''_{1,1}=1=A_{1,1}$, then
  we have $l$ is odd and $A=\sum_{i=1}^{l} v_{i} \otimes v_{i}$ as promised.
  Otherwise $l$ is even. Let $e=(1\ 0\ \cdots\ 0) \in \F{2}^{m+1}$. Then
  $A=\sum_{i=1}^{l} v_i \otimes v_i + e \otimes e$ gives the desired decomposition.
\end{proof}

The following lemma summarizes the discussion at the beginning of this section
and relates odd superposition complexity with low-rank pseudo-quadratic matrices.
\begin{lemma}\label{lem:rankoddcov}
  Let $q_1(x)=0,\cdots,q_s(x)=0$ be a set of $s$ quadratic equations
  on variable $x_1,\cdots,x_m$, and let $Q_1,\cdots,Q_s \in \F{2}^{(m+1) \times (m+1)}$
  be their corresponding matrix forms.
  Suppose there is a pseudo-quadratic matrix $A \in \F{2}^{(m+1) \times (m+1)}$
  such that $\rank(A) \le k$ and for all $i \in [s]$, $\langle Q_i,A \rangle=0$,
  then there exists integer $l<3k/2+1$ and $l$ vectors $a^{(1)},\cdots,a^{(l)} \in \F{2}^{m+1}$,
  such that $A=\sum_{i=1}^{l} a^{(i)} \otimes a^{(i)}$,
  where for all $j \in [l]$, we have $a^{(j)}_1=1$ and that $D_1(a^{(j)})$ is in the column
  space of $D_1(A)$.
  This implies that the assignments $D_1(a^{(1)}),\cdots,D_1(a^{(l)})$
  satisfy all equations $q_1(x)=0,\cdots,q_s(x)=0$ in superposition.
\end{lemma}
\begin{proof}
  Apply Lemma \ref{lem:odddecomp} to $A$, and let $v_1,\cdots,v_l$ be the vectors
  we get, with $v_{i1}=1$ for $i \in [l]$, and $A=\sum_{i \in [l]} v_i \otimes v_i$.
  We now verify that $D_1(v_1),\cdots,D_1(v_l)$ satisfy all equations in superposition.

  Consider equation $i$ for $i \in [s]$. We have
  \begin{align*}
    0=\langle Q_i,A \rangle \ =&\ 
    \sum_{i=1}^{l} \langle Q_i, v_i \otimes v_i \rangle \\
    =&\  \sum_{i=1}^{l} q_i(v_i)\,.
  \end{align*}
  By definition, we have that $v_1,\cdots,v_l$ satisfy $q_i$ in superposition.
\end{proof}

\subsection{\LC}
The starting point of our reduction is the \LC hardness
obtained from \ksat{3} instances. 
We use \LC instances obtained by 
applying the PCP Theorem \cite{arorapcp1,arorapcp2}
and the Parallel Repetition Theorem \cite{raz}. 
The exact formulation below is from \cite{shortcodecol}.
\begin{definition}\label{def:3satlc}
  Let $\phi$ be a \ksat{3} instance with $X$ as the set of variables
  and $\mathcal{C}$ the set of clauses. The $r$-repeated \LC
  instance $\mathcal{L}(r,\phi)$ is specified by:
  \begin{itemize}
    \item A bipartite graph $G=(U,V,E)$, where $V:=\mathcal{C}^r$
      and $U:=X^r$.
    \item Label set for $U$, denote by $L:=\bool^r$, and label
      set for $V$, denote by $R:=\bool^{3r}$.
    \item There is an edge $\{u,v\} \in E$ if for each $i \in [r]$,
      $u_i$ is a variable appearing in clause $v_i$.
    \item For edge $\{u,v\}$, the constraint $\pi_{uv}:\bool^{3r} \to \bool^r$
      is the projection of the assignment of the $3r$ clause variables in $v$
      to the assignment of the $r$ variables in $u$.
    \item For each $v \in V$, there is a set of $r$ functions
      $\{f^v_i:\bool^{3r} \to \bool^r\}_{i \in [r]}$,
      such that $f^v_i(a)=0$ if and only if the assignment $a$ satisfies
      the clause $v_i$. Note that each $f^v_i$ depends only on $3$ entries of $a$.
  \end{itemize}
  A labeling $\sigma:U \to L,V \to R$ satisfies an edge $\{u,v\}$ iff
  $\pi_{uv}(\sigma(V))=\sigma(U)$, and $\sigma(V)$ satisfies all clauses
  in $v$. The value of $\mathcal{L}(r,\phi)$ is the maximum fraction
  of edges that can be simultaneously satisfied by any labeling.
\end{definition}

We have the following hardness result for \LC.
\begin{theorem}\label{thm:3satlc}
  Given a \ksat{3} instance $\phi$ on $n$ variables and $r \in \N$, 
  there is an algorithm that constructs $\mathcal{L}(r,\phi)$
  in time $n^{O(r)}$, and that the output \LC instance has the following
  properties:
  \begin{itemize}
    \item If $\phi$ is satisfiable, then the value of $\mathcal{L}(r,\phi)$ is $1$.
    \item If $\phi$ is unsatisfiable, then the value of $\mathcal{L}(r,\phi)$
      is at most $2^{-\varepsilon_0 r}$, for some universal constant 
      $\varepsilon_0 \in (0,1)$.
  \end{itemize}
\end{theorem}

In our construction of \LC instance with matrix labels, we need to use
the following Parallel Repetition theorem from Rao \cite{parreprao}, which applies
to projection games (\LC), with the advantage that the rate at which
the soundness decreases is independent of the label size of the original instance.
\begin{theorem}[Parallel Repetition \cite{parreprao}]\label{thm:parrep}
  There is a universal constant $\alpha>0$, such that 
  for a \LC instance $\Psi$,
  if $\Opt(\Psi) \le 1-\varepsilon$, then 
  $\Opt(\Psi^n) \le (1-\varepsilon/2)^{\alpha\varepsilon n}$.
\end{theorem}

\subsection{\LDLC}
In this section, we review the basics of \LDLC. The formulation here
is from \cite{dinurguruswamishortcode} and \cite{shortcodecol}.
Towards the end of this section, we prove a key lemma that we 
will use for proving our superposition hardness results.

For a positive integer $m$, denote by $\Pm{m}$ the vector space of
$m$-variable functions $\F{2}^m \to \F{2}$. For $f,g \in \Pm{m}$,
let $\Delta(f,g)$ be the Hamming distance between $f$ and $g$.
For a subset of functions $\mathcal{F} \subseteq \Pm{m}$,
the distance between $g$ and $\mathcal{F}$ is defined as
$\Delta(g,\mathcal{F})=\min_{f \in \mathcal{F}}\Delta(f,g)$.

We define the following dot product on $\Pm{m}$.
\begin{definition}[Dot Product]
  For $f,g \in \Pm{m}$, the dot product is defined as 
  $\langle f,g \rangle=\sum_{x \in \F{2}^m} f(x)g(x)$.
\end{definition}

Denote by $\Pmd{m}{d}$ be the space of functions with degree at most $d$.
For a subspace $\mathcal{A} \subseteq \Pmd{m}{d}$, denote its dual
by $\mathcal{A}^{\bot} = \{g \in \Pm{m} \mid \forall f \in \mathcal{A}, \langle f,g \rangle = 0 \}$.
It is well known that $\Pmd{m}{d}^{\bot}=\Pmd{m}{m-d-1}$.

For $\beta \in \Pm{m}$, denote by $\supp(\beta)$ the support of $\beta$,
that is $\supp(\beta)=\{x \mid \beta(x)=1\}$. Define $\weight(\beta)=|\supp(\beta)|$.

\begin{definition}[\LDLC]
  The \LDLC encoding for an $m$-bit string $a \in \F{2}^m$ is a function
  $A_a:\Pmd{m}{d} \to \F{2}$, defined as
  $A_a(g)=g(a)$, for all $g \in \Pmd{m}{d}$.
\end{definition}

\begin{definition}[Character Set]
  For $\beta \in \Pm{m}$, define the corresponding \emph{character function}
  $\chi_{\beta}:\Pmd{m}{d} \to \R$ as $\chi_{\beta}(f)=(-1)^{\langle \beta,f \rangle}$.

  Define the \emph{character set} $\Lambda_{m,d}$ to be the set of functions $\beta \in \Pm{m}$
  which are minimum weight functions in the cosets of $\Pm{m}/\Pmd{m}{d}^{\bot}$, where ties
  are broken arbitrarily.
\end{definition}

We have the following result about the character set and
the ``Fourier decomposition'' for functions $\Pmd{m}{d} \to \R$ from \cite{dinurguruswamishortcode}.
\begin{lemma}
  The following statements hold:
  \begin{itemize}
    \item For any $\beta,\beta' \in \Pm{m}$, $\chi_{\beta}=\chi_{\beta'}$
      if and only if $\beta-\beta' \in \Pmd{m}{d}^{\bot}$.
    \item For $\beta \in \Pmd{m}{d}^{\bot}$, $\chi_{\beta}$ is the constant $1$ 
      function.
    \item For any $\beta$, there exists $\beta'$, such that
      $\beta-\beta' \in \Pmd{m}{d}^{\bot}$, and 
      $|\supp(\beta')|=\Delta(\beta,\Pmd{m}{d}^{\bot})$.
      We call such $\beta'$ the minimum support function for the coset
      $\beta+\Pmd{m}{d}^{\bot}$.
    \item The characters in the character set $\Lambda_{m,d}$ form an orthonormal 
      basis under the inner product $\langle A,B \rangle=\E_{f \in \Pmd{m}{d}}[A(f)B(f)]$.
    \item Any function $A:\Pmd{m}{d} \to \R$ can be uniquely decomposed as
      \[
      A(g) = \sum_{\beta \in \Lambda_{m,d}} \widehat{A}_{\beta} \chi_{\beta}(g)\,.
      \]
    \item Parseval's identity: For any $A:\Pmd{m}{d} \to \R$,
      $\sum_{\beta \in \Lambda_{m,d}} \widehat{A}_{\beta}^2 = \E_{f \sim \Pmd{m}{d}}[A(f)^2]$.
  \end{itemize}
\end{lemma}

The following lemma relates characters from different domains related by coordinate projections
and is from \cite{dinurguruswamishortcode}.
\begin{lemma}
  Let $n \le m$, and $S \subseteq [m]$ with $|S|=n$, and let $\pi:\F{2}^m \to \F{2}^n$
  be a projection, mapping $x \in \F{2}^m$ to $x|_S \in \F{2}^n$. Then 
  for $f \in \Pmd{n}{d}$ and $\beta \in \Pm{m}$, we have
  \[
  \chi_{\beta}(f \circ \pi) = \chi_{\pi_2(\beta)}(f)\,,
  \]
  where $\pi_2(\beta)(y)=\sum_{x \in \pi^{-1}(y)}\beta(x)$.
\end{lemma}

Like in the classical \LGCD reductions, we enforce special structures on the 
tables. This is a technique known as \emph{folding}. The following properties of the Fourier coefficients
of folded functions were also studied in \cite{dinurguruswamishortcode}.

\begin{definition}
  A table $A:\Pmd{m}{d} \to \R$ is folded 
  if for any $f \in \Pmd{m}{d}$, we have $A(f+1)=-A(f)$.
\end{definition}
\begin{lemma}\label{lem:shortcodeodd}
  If $A:\Pmd{m}{d} \to \R$ is folded, then
  for any $\alpha$ such that $\widehat{A}_{\alpha} \ne 0$,
  we have $\sum_{x \in \F{2}^m} \alpha(x)=1$. In particular, we have
  $\supp(\alpha) \ne \emptyset$.
\end{lemma}
%\begin{proof}
%  By definition, we have
%  \begin{align*}
%    \widehat{A}_{\alpha} ~=&~
%    \E_{f \sim \Pmd{m}{d}}[A(f) \chi_{\alpha}(f)] 
%    = \E_{f \sim \Pmd{m}{d}}[A(f+1) \chi_{\alpha}(f+1)] \\
%    =&~ -\E_{f \sim \Pmd{m}{d}}[A(f) \chi_{\alpha}(f) \chi_{\alpha}(1)] \\
%    =&~ (-1)^{\sum_{x \in \F{2}^m} \alpha(x)+1} \widehat{A}_{\alpha}\,.
%  \end{align*}
%  Therefore if we have $\widehat{A}_{\alpha} \ne 0$, then we must have
%  $\sum_{x \in \F{2}^m} \alpha(x)=1$.
%\end{proof}

\begin{definition}
  Let $q_1,\cdots,q_k \in \Pmd{m}{3}$, and let
  \[
  J(q_1,\cdots,q_k):=\left\{ \sum_{i} r_i q_i \mid r_i \in \Pmd{m}{d-3} \right\}\,.
  \]
  We say that a function $A:\Pmd{m}{d} \to \R$ is \emph{folded over $J$} if
  $A$ is constant over cosets of $J$ in $\Pmd{m}{d}$.
\end{definition}
The following lemma shows that a function folded over $J$ does not have
weight on small support characters that are non-zero on $J$.
\begin{lemma}\label{lem:shortcodefoldsubspace}
  Let $\beta \in \Pm{m}$ be such that $\weight(\beta) < 2^{d-3}$, and there exists
  some $i \in [k]$ and $x \in \supp(\beta)$ with $q_i(x) \ne 0$.
  Then if $A:\Pmd{m}{d} \to \R$ is folded over $J$, then $\widehat{A}_{\beta}=0$.
\end{lemma}
In the actual reduction, $q_1,\cdots,q_k$ will be the set of functions
associated with vertices in the \LC instance, as described in Definition
\ref{def:3satlc}.

In \cite{dinurguruswamishortcode}, Dinur and Guruswami proved the following 
theorem about Reed-Muller codes over $\F{2}$.
\begin{theorem}
  Let $d$ be a multiple of $4$. If $\beta \in \Pm{m}$ is such that
  $\Delta(\beta,\Pmd{m}{d}) \ge 2^{d/2}$, then 
  \[
  \E_{g \sim \Pmd{m}{d/4}}\left[ 
  \left| \E_{h \sim \Pmd{m}{3d/4}} [\chi_{\beta}(gh)]
  \right|
  \right] \le 2^{-4 \cdot 2^{d/4}}\,.
  \]
\end{theorem}
Note that $\chi_{\beta}(gh)=(-1)^{\langle \beta g,\beta h\rangle}$. 
The key lemma we will now prove is a generalization of the above theorem.
The setting is that we have an additional $t$ functions
$A_1,\cdots,A_t:\Pmd{m}{d} \to \F{2}$. We show that as long as $t$ is small
compared to $2^{d/2}$, the expectation 
$\E_{g,h}[(-1)^{\langle \beta g,\beta h\rangle+\sum_{i=1}^{t} A_i(g)A_i(h)}]$
is still close to $0$ for arbitrary $A_1,\cdots,A_t$.

We use some of the key steps in \cite{dinurguruswamishortcode}.
\begin{definition}
  For $\beta$ and $k \le d$, define
  \[
  B^{(m)}_{d,k}:=\left\{ g \in \Pmd{m}{k} \mid \beta g \in \Pmd{m}{m-d-1+k} \right\}\,.
  \]
  Note that $B^{(m)}_{d,k}$ is a subspace of $\Pmd{m}{k}$.
\end{definition}

For positive integers $d,k$, define $\Phi_{d,k}:\N \to \N$ as follows: if $d<k$, then 
$\Phi_{d,k}$ is identically $0$, otherwise
\[
\Phi_{d,k}(D)=\min_{\substack{m>d \\ \beta \in \Pm{m}: \Delta(\beta,P(m,m-d-1)) \ge D}}
\left\{ \lindim(P(m,k))-\lindim(B^{(m)}_{d,k}(\beta)) \right\}\,.
\]

The following two claims are from \cite{dinurguruswamishortcode}, 
which serve as the basis step and induction step
for their lower-bound for $\Phi_{d,k}(D)$.
\begin{claim}
  For $d \ge k$ and $D \ge 1$, $\Phi_{d,k}(D) \ge 1$.
\end{claim}
\begin{claim}
  For all $d \ge k$ and $40<D<2^d$, $\Phi_{d,k}(D) \ge \Phi_{d-1,k}(D/4)+\phi_{d-1,k-1}(D/4)$.
\end{claim}
For $D=2^{d-4}=4^{d/2-2}$ and $k=d/2$, applying the above for a depth of $d/2-4$, reducing $D$ from
$4^{d/2-2}$ to $16$, we have $\Phi_{d,d/2}(2^{d-4}) \ge 2^{d/2-4}$.
This gives the following theorem.
\begin{theorem}\label{thm:maincodim}
  For all integers $m,d$ such that $m>d>0$ and $4|d$, if $\beta:\F{2}^m \to \F{2}$
  has distance more than $2^{d-4}$ from $\Pmd{m}{m-d-1}$, then the subspace
  $B^{(m)}_{d,d/2}(\beta)$ (as a subspace of $\Pmd{m}{d/2}$) has codimension at least
  $2^{d/2-4}$.
\end{theorem}
We remark that Dinur and Guruswami used different degree parameters 
in \cite{dinurguruswamishortcode} for their application. Otherwise, the above
theorem is the same as in \cite{dinurguruswamishortcode}.

We are now ready to prove the main theorem of this section.
\begin{theorem}\label{thm:lowdegreebias}
  Let $\beta:\F{2}^m \to \F{2}$ be a polynomial with distance more than $2^{d-4}$
  from $\Pmd{m}{m-d-1}$. Let $t \in \N$ and $A_1,\cdots,A_t:\Pmd{m}{d/2} \to \F{2}$
  be some arbitrary $t$ functions. Let $\mu$ be the uniform distribution
  on $\Pmd{m}{d/2}$. Then
  \begin{align*}
    &~ \E_{g,h \sim \mu}\left[ \chi_{\beta}(gh)\cdot (-1)^{\sum_{i=1}^{t}A_i(g)A_i(h)} \right] \\
    =&~
    \E_{g,h \sim \mu}\left[ (-1)^{\langle \beta g,\beta h \rangle+\sum_{i=1}^{t}A_i(g)A_i(h)} \right]
    \le 2^{-(2^{d/2-4}-t)/2}\,.
  \end{align*}
\end{theorem}
\begin{proof}
  Denote by $\mathcal{W}$ the quotient space $\Pmd{m}{d/2} / B^{(m)}_{d,d/2}(\beta)$.
  By Theorem \ref{thm:maincodim}, we have 
  $w:=\lindim(\mathcal{W})=\codim(B^{(m)}_{d,d/2}(\beta)) \ge 2^{d/2-4}$.

  The expectation we are considering can be written as
  \begin{align}
    \E_{g_0,h_0 \sim \mathcal{W}}
    \E_{\substack{g: g-g_0 \in B^{(m)}_{d,d/2}(\beta) \\ 
    h:h-h_0 \in B^{(m)}_{d,d/2}(\beta)}}
    \left[ (-1)^{\langle \beta g,\beta h\rangle+\sum_{i=1}^{t}A_i(g)A_i(h)} \right]\,.
    \label{eq:lowdegexp1}
  \end{align}

  Consider $f \in \Pmd{m}{d/2}$ and $g \in B^{(m)}_{d,d/2}(\beta)$. 
  We have $\langle \beta f, \beta g \rangle =\langle \beta g, f\rangle=0$,
  because $f \in \Pmd{m}{d/2}$ and $\beta g \in \Pmd{m}{m-d/2-1}=\Pmd{m}{d/2}^{\bot}$.
  This allows us to define ``dot product'' between elements in $\mathcal{W}$. In particular,
  for any $f,f',g,g' \in \Pmd{m}{d/2}$ such that $f-f',g-g' \in B^{(m)}_{d,d/2}(\beta)$, we have
  \begin{align*}
    &~ \langle \beta f',\beta g' \rangle \\
    =&~ \langle \beta f',\beta g' \rangle+
    \langle \beta(f-f'),\beta g' \rangle + \langle \beta f',\beta (g-g') \rangle
    +\langle \beta(f-f'),\beta(g-g') \rangle \\
    =&~ \langle \beta f,\beta g \rangle \,.
  \end{align*}
  This means that taking any representative from $\mathcal{W}$ will give the same result
  for this ``dot product''.

  We can thus further rewrite the expectation as
  \begin{equation}
    (\ref{eq:lowdegexp1}) = 
    \E_{g_0,h_0 \sim \mathcal{W}}\left[
    (-1)^{\langle \beta g_0, \beta h_0 \rangle}
    \E_{\substack{g:g-g_0 \in B^{(m)}_{d,d/2}(\beta) \\ 
    h:h-h_0 \in B^{(m)}_{d,d/2}(\beta)}}
    \left[ (-1)^{\sum_{i=1}^{t}A_i(g)A_i(h)} \right]\right]\,.
    \label{eq:lowdegexp2}
  \end{equation}

  Consider the matrix $M \in \R^{2^{w+t} \times 2^{w+t}}$, where
  the rows and columns are indexed by a pair $(f_0,a)$ where $f_0 \in \mathcal{W}$
  and $a \in \F{2}^t$, and the entries are 
  \[
  M_{(f_0,a),(g_0,b)} = (-1)^{\langle \beta f_0, \beta g_0 \rangle+\sum_{i=1}^{t} a_i b_i}\,.
  \]
  Define vector $u \in \R^{2^{w+t}}$ as
  \[
  u_{f_0,a} = \Pr_{g \sim \Pmd{m}{d/2}}
  \left[
  g-f_0 \in B^{(m)}_{d,d/2}(\beta)
  \land \forall i \in [t], A_i(\mathbf(g))=a_i
  \right]\,.
  \]
  Since in (\ref{eq:lowdegexp2}), $g$ and $h$ are sampled independently,
  we can verify that the expectation in (\ref{eq:lowdegexp2}) is exactly $u^T M u$.
  Moreover, since $g$ is chosen uniformly random from $\Pmd{m}{d/2}$, the probability
  that $g-f_0 \in B^{(m)}_{d,d/2}(\beta)$ is exactly $2^{-w}$, thus
  all entries in $u$ have value at most $2^{-w}$, and therefore
  $\|u\|_2 \le 2^{-w/2}$.

  We finish the proof by studying the spectrum of $M$.
  Observe that $M$ can be written as the tensor product of
  a $2^w \times 2^w$ matrix and a $2^t \times 2^t$ matrix as follows.
  Define $W \in \R^{2^w \times 2^w}$ as
  \[
  W_{f_0,g_0} = (-1)^{\langle \beta f_0, \beta g_0 \rangle}\,,
  \]
  for $f_0,g_0 \in \mathcal{W}$. 
  Define $H \in \R^{2^t \times 2^t}$ as
  \[
  H_{a,b} = (-1)^{\sum_{i=1}^{t} a_i b_i}\,.
  \]
  We can easily verify that $M=W \otimes H$.

  The matrix $H$ satisfies $H H^T = 2^t \cdot I$, where $I$ is the identity matrix,
  therefore we have that the eigenvalues of $H$ all have absolute value exactly $2^{t/2}$.
  For the spectrum of $W$, let $f_0,g_0 \in \mathcal{W}$ be two rows of $W$.
  Consider the dot product of row $f_0$ and $g_0$ of matrix $W$
  \[
  W_{f_0}^T W_{g_0} 
  = \sum_{h_0 \in \mathcal{W}} (-1)^{\langle \beta(f_0+g_0),\beta h_0 \rangle}
  = \sum_{h_0 \in \mathcal{W}} (-1)^{\langle \beta(f_0+g_0), h_0 \rangle}\,.
  \]
  The above sum is $2^w$ if $\beta(f_0+g_0) \in \Pmd{m}{m-d/2-1}$,
  or in other words $f_0$ and $g_0$ belong to the same coset in
  $\mathcal{W}$,
  and otherwise the sum is $0$. Hence we have $W W^T = 2^w \cdot I$, and thus
  the eigenvalues of $W$ all have absolute value $2^{w/2}$.
  We conclude that the tensor product matrix $M=W \otimes H$ has eigenvalues with absolute value
  $2^{(w+t)/2}$.

  We can now upper-bound the absolute value of the expectation by
  $|u^T M u| \le 2^{(w+t)/2} \cdot \|u\|_2^2 = 2^{-(w-t)/2}$.
\end{proof}

\section{Superposition Hardness for Gap TSA}\label{sec:octsa}
Let $b$ be some large integer parameter.
The Tri-Sum-And (\tsa) predicate is a predicate
on $5$ $\F{2}$-variables defined as follows
\[
\tsa(x_1,\cdots,x_5) = 1+x_1+x_2+x_3 + x_4 x_5\,.
\]
From the definition, we can see that \tsa instances are
systems of quadratic equations, each involving exactly $5$ 
$\F{2}$-variables.

The predicate was studied in \cite{hkperfect} as a starting
point of an efficient PCP construction.
For the predicate itself, H{\aa}stad and Khot proved that
it is approximation resistant on satisfiable instances.

In this section, we prove a superposition hardness result for \tsa.
\begin{theorem}\label{thm:octsa}
  There is a reduction that takes as input a \ksat{3} instance of
  size $n$, and outputs a \tsa instance of size $n^{O(b\log\log n)}$
  with the following properties:
  \begin{itemize}
    \item If the \ksat{3} instance is satisfiable, then there is an 
      assignment that satisfies all TSA constraints.
    \item If the \ksat{3} instance is unsatisfiable, then for any
      odd integer $t<(\log n)^{b}$, and any $t$ assignments,
      at most a $15/16$ fraction of the \tsa constraints
      are satisfied in superposition.
  \end{itemize}
  The reduction runs in time $n^{O(b \log\log n)}$.
\end{theorem}
\begin{proof}
  The reduction follows a similar approach as a typical inapproximability
  hardness reduction.

  Given a \ksat{3} instance, we apply Theorem \ref{thm:3satlc}
  with soundness $1/(1000(\log n)^{2b})$ to get a \LC instance.
  This gives the parameter $r=(2b \log\log n+O(1))/\varepsilon_0$, where $\varepsilon_0$ is
  some universal constant.
  The vertex set of the bipartite graph has size $n^{O(b \log\log n)}$,
  and the label sets are $L=\bool^r$ and $R=\bool^{3r}$.
  Let $d=\Theta(b \log\log n)$ be such that $2^{d/2-4} \approx (\log n)^b+3$.
  This implies also that $2^{d} \approx 256 (\log n)^{2b}$.

  For each $u \in U$ and $v \in V$, we expect functions
  $f_u :\Pmd{r}{d} \to \boolean$ and $g_v:\Pmd{3r}{d} \to \boolean$.
  We assume that all functions are folded over constant.
  The entries of the functions correspond to variables of some \tsa instance.
  Therefore the number of variables in the output instance is
  $n^{O(b \log\log n)} \cdot (3r)^{(1+o(1))d} = n^{O(b\log\log n)}$,
  and the number of constraints is polynomial in the number of variables.

  Consider the following test:
  \begin{enumerate}
    \item Sample random edge $e=\left\{ u_1,u_2 \right\} \sim E$.
      Let $\pi$ be the projection on the edge,
      and let $f$ and $g$ be the functions associated with $u_1$ and
      $u_2$.
    \item Sample uniformly random query
      $x \sim \Pmd{r}{d}$, $y \sim \Pmd{3r}{d}$, and $v,w \sim \Pmd{3r}{d/2}$.
    \item Construct query $z:=x \circ \pi+y+vw \in \Pmd{3r}{d}$.
    \item Accept iff 
      $f(x)g(y)g(z)(g(v) \land g(w)) = 1$,
      where $\land$ here denotes the binary operator that evaluates to
      $-1$ when both operands are $-1$, and $1$ otherwise.
  \end{enumerate}
  The completeness is straightforward. In this case, the \LC instance
  has a perfect labeling. Setting the functions to be the \LDLC
  encoding of the labels gives an assignment that satisfies all
  \tsa constraints.

  In the soundness case, there exists some $t<(\log n)^b$ assignments
  that satisfy in superposition a $15/16$ fraction of the constraints.
  That is, for each $u_1 \in U$ and $u_2 \in V$,
  there are $t$ functions that are folded over constant,
  $f^{(1)},\cdots,f^{(t)}:\Pmd{r}{d} \to \boolean$ and 
  $g^{(1)},\cdots,g^{(t)}:\Pmd{3r}{d} \to \boolean$
  such that over random sample of edges $\{u_1,u_2\}$
  and queries $x,y,z,v,w$, with probability
  at least $15/16$,
  the number of $i \in [t]$ such that
  $f^{(i)}(x) g^{(i)}(y) g^{(i)}(z) (g^{(i)}(v) \land g^{(i)}(w))=1$
  is odd.  
  By an averaging argument, we have that for at least $3/4$ of the edges,
  over random sample of queries, the above holds with probability at least $3/4$.
  Call such an edge \emph{good}.

  We assume that the functions are folded in the same way. Recall that when applying
  folding, we partition the domain of the functions into equivalence classes, define
  the function value in one of the equivalence classes, and then extend to the full
  domain by adding appropriate constants. For our reduction, we identify one equivalence
  class for each vertex, and the $t$ functions associated with it supply value
  only for that equivalence class. This is to make sure $f^{(1)},\cdots,f^{(t)}$
  and $g^{(1)},\cdots,g^{(t)}$ corresponds exactly to $t$ assignments in superposition.
  
  Fix a good edge for now, and we drop the subscripts $u_1$ and $u_2$.
  Then we have the following
  \[
  \frac{1}{2}+\frac{1}{2}\E_{x,y,z,v,w}\left[ \prod_{i=1}^{t} 
  \left( f^{(i)}(x)g^{(i)}(y)g^{(i)}(z)(g^{(i)}(v) \land g^{(i)}(w)) \right) \right]
  \ge \frac{3}{4}\,,
  \]
  or 
  \[
  \E_{x,y,z,v,w}\left[ \prod_{i=1}^{t} 
  \left( f^{(i)}(x)g^{(i)}(y)g^{(i)}(z)(g^{(i)}(v) \land g^{(i)}(w)) \right) \right]
  \ge \frac{1}{2}\,.
  \]

  Let $f'=\prod_{i=1}^{t} f^{(i)}$, and $g'=\prod_{i=1}^{t} g^{(i)}$. Since $t$ is odd,
  we have that $f'$ and $g'$ are both folded over constant. Taking the Fourier expansion
  of $f'$ and $g'$, we have the following
  \begin{align*}
    \frac{1}{2} \ \le&\  \E_{x,y,z,v,w}\left[ \prod_{i=1}^{t} 
    \left( f^{(i)}(x)g^{(i)}(y)g^{(i)}(z)(g^{(i)}(v) \land g^{(i)}(w)) \right) \right] \\
    =&\  \E\left[ f'(x) g'(y) g'(z) \prod_{i=1}^{t}(g^{(i)}(v) \land g^{(i)}(w))
    \right] \\
    =&\ 
    \sum_{\substack{\alpha \in \Lambda_{r,d} \\ \beta_1,\beta_2 \in \Lambda_{3r,d}}}
    \widehat{f'}_{\alpha}\widehat{g'}_{\beta_1}\widehat{g'}_{\beta_2} \\
    &\  \E_{x,y,z,v,w}\left[ 
    \chi_{\alpha}(x)\chi_{\beta_1}(y)\chi_{\beta_2}(x \circ \pi+y+vw)
    \prod_{i=1}^{t}(g^{(i)}(v) \land g^{(i)}(w))
    \right] \\
    =& \ 
    \sum_{\beta \in \Lambda_{3r,d}} \widehat{f'}_{\pi_2(\beta)}\widehat{g'}_{\beta}^2
    \E_{vw}\left[\chi_{\beta}(vw) \prod_{i=1}^{t}(g^{(i)}(v) \land g^{(i)}(w))\right]\,.
  \end{align*}
  Applying Cauchy-Schwarz and using Parseval, we have
  \begin{align*}
    \frac{1}{4} \ \le&\ 
    \left( \sum_{\beta \in \Lambda_{3r,d}} \widehat{g'}_{\beta}^2 \right)
    \left(  
    \sum_{\beta \in \Lambda_{3r,d}} \widehat{f'}_{\pi_2(\beta)}^2\widehat{g'}_{\beta}^2
    \E_{vw}\left[\chi_{\beta}(vw) \prod_{i=1}^{t}(g^{(i)}(v) \land g^{(i)}(w))\right]^2
    \right) \\
    =&\ 
    \sum_{\beta \in \Lambda_{3r,d}:\weight(\beta) \le 2^{d-4}} 
    \widehat{f'}_{\pi_2(\beta)}^2\widehat{g'}_{\beta}^2
    \E_{vw}\left[\chi_{\beta}(vw) \prod_{i=1}^{t}(g^{(i)}(v) \land g^{(i)}(w))\right]^2+\\
    &\  \sum_{\beta \in \Lambda_{3r,d}:\weight(\beta) > 2^{d-4}} 
    \widehat{f'}_{\pi_2(\beta)}^2\widehat{g'}_{\beta}^2
    \E_{vw}\left[\chi_{\beta}(vw) \prod_{i=1}^{t}(g^{(i)}(v) \land g^{(i)}(w))\right]^2\,.
  \end{align*}
  For the terms where $\weight(\beta)>2^{d-4}$, we apply Theorem \ref{thm:lowdegreebias} to get
  \[
  \left|\E_{vw}\left[\chi_{\beta}(vw) \prod_{i=1}^{t}(g^{(i)}(v) \land g^{(i)}(w))\right]\right|
  \le 2^{-(2^{d/2-4}-t)/2} \,,
  \]
  and therefore 
  \begin{align*}
    &\  \sum_{\beta \in \Lambda_{3r,d}:\weight(\beta) > 2^{d-4}} 
    \widehat{f'}_{\pi_2(\beta)}^2\widehat{g'}_{\beta}^2 \\
    &\  \E_{vw}\left[\chi_{\beta}(vw) \prod_{i=1}^{t}(g^{(i)}(v) \land g^{(i)}(w))\right]^2
    \le 2^{-(2^{d/2-4}-t)} < \frac{1}{8}\,.
  \end{align*}
  This gives us
  \begin{align*}
    &\  \sum_{\beta \in \Lambda_{3r,d}:\weight(\beta) \le 2^{d-4}} 
    \widehat{f'}_{\pi_2(\beta)}^2\widehat{g'}_{\beta}^2 \\
    \ge&\ 
    \sum_{\beta \in \Lambda_{3r,d}:\weight(\beta) \le 2^{d-4}} 
    \widehat{f'}_{\pi_2(\beta)}^2\widehat{g'}_{\beta}^2
    \E_{vw}\left[\chi_{\beta}(vw) \prod_{i=1}^{t}(g^{(i)}(v) \land g^{(i)}(w))\right]^2
    \ge \frac{1}{8}\,.
  \end{align*}

  Let $\{u_1,u_2\}$ be a good edge. Consider the following labeling strategy:
  for $u_1$, pick $\alpha$ with probability $\widehat{f'}_{\alpha}^2$ and pick
  a random label from $\supp(\alpha)$, and for $u_2$,
  pick $\beta$ with probability $\widehat{g'}_{\beta}^2$ and pick a random label 
  from $\supp(\beta)$. 
  The procedure is well defined because $f'$ and $g'$ are all folded, and thus
  by Lemma \ref{lem:shortcodeodd},
  $\supp(\alpha)$ and $\supp(\beta)$ are nonempty. Also, for $\beta$ such that
  $\weight(\beta) \le 2^{d-4}<2^{d-3}$, by Lemma \ref{lem:shortcodefoldsubspace},
  the assignments in $\supp(\beta)$ all satisfy the clauses in $u_2$.
  Then the probability that the labeling of $u_1$ and $u_2$ satisfies
  the projection constraint on a good edge $\{u_1,u_2\}$ is at least 
  $\frac{1}{2^{d-4}} \sum_{\beta:wt(\beta) \le 2^{d-4}} 
  \widehat{f'}_{\pi_2(\beta)}^2 \widehat{g'}_{\beta}^2 
  \ge 1/(8 \cdot 2^{d-4}) > 1/(100 (\log n)^{2b})$.
  Since there are at least a $3/4$ fraction of good edges, overall the labeling satisfies
  more than $(3/4) \cdot (1/(100 (\log n)^{2b}))>1/(1000 (\log n)^{2b})$,
  contradicting the fact that in the soundness case the \LC instance does not have
  labeling with value larger than $1/(1000(\log n)^{2b})$. This completes the proof.
\end{proof}

\section{Label Cover with Matrix Labels}
We now use Theorem \ref{thm:octsa} to construct a \LC
instance with properties similar to that in \cite{kssp}.
The proof follows closely the approach in \cite{ksspeccc}, Section 5 -- 7.

We first give an analogue of Theorem \ref{thm:octsa} over a large field
$\F{q}$ of characteristic $2$.

\begin{theorem}\label{thm:octsa-ext}
  Let $q=2^{(\log n)^{b+4}}$.
  There is a reduction that takes as input a \ksat{3} instance of
  size $n$, and outputs a system $\mathcal{I}$ of quadratic equations over $\F{q}$
  of size $n^{O(b \log\log n)}$ with the following properties:
  \begin{itemize}
    \item Each quadratic equation in $\mathcal{I}$ involves
      exactly $5$ variables.
    \item If the \ksat{3} instance is satisfiable, then there is an 
      assignment that satisfies all equations in $\mathcal{I}$.
    \item If the \ksat{3} instance is unsatisfiable, then for any
      integer $t<(\log n)^{b}$, and any $t$ assignments,
      at most a $15/16$ fraction of the equations in $\mathcal{I}$
      are satisfied in superposition.
  \end{itemize}
  The reduction runs in time $n^{O(b \log\log n)}$.
\end{theorem}
The proof is almost identical to Theorem 3.4 and Theorem 5.3 of \cite{ksspeccc},
where they proved that a $(t \cdot \log q)$-superposition hardness with gap $\delta$
for systems of quadratic equations over $\F{2}$ implies
$t$-superposition hardness with the same gap for systems of quadratic equations over
$\F{q}$, where $q=2^r$ and $\F{q}$ is an extension field of $\F{2}$.

We now construct \LC with matrix labels from the above theorem.
The construction is via a Point vs. Ruled-surface test. 
The test is very similar to the one in \cite{ksspeccc}. The main difference
is here we use the low error version of the Low Degree Test,
instead of the one used by Khot and Saket in \cite{ksspeccc}. The following
theorem appears as Theorem 5.1 in \cite{arorathesis},
which follows from the works of Arora and Sudan \cite{arorasudanlowdegree}
and Rubinfeld and Sudan \cite{rubinfeldsudanlowdegree}.

\begin{theorem}\label{thm:lowdegreelowerror}
  Let $\alpha \le 10^{-4}$, and $d$ and $m$ be positive integers, and
  $\F{q}$ be a field with $q>100 d^3 m$. 
  Let $f:\F{q}^m \to \F{q}$ be a function, and for every line in $\F{q}^m$, let
  $f_l$ be a univariate degree $d$ polynomial. For a point $x$ on the line $l$,
  we denote by $f_l(x)$ the value given by $f_l$ at the point $x$.

  Suppose that over the choice of a random line $l$ and a random point
  $x \in l$, we have $\Pr_{l,x}[f(x)=f_l(x)] \ge (1-\alpha)$,
  then there is a multivariate polynomial of total degree at most $d$ which agrees
  with $f$ on at least $(1-2\alpha)$ fraction of the points.
\end{theorem}

Let $\mathcal{I}$ be the instance produced by Theorem \ref{thm:octsa-ext},
and let $N=n^{O(b \log\log n)}$ be the number of variables in $\mathcal{I}$.
Let $m=\lceil \log N / \log\log N \rceil = O(\log n)$, 
and $h=\lceil \log N \rceil=O(\log n \log\log n)$, so that
$h^m \ge N$, and let $d:=m(h-1)=O(\log^2 n \log\log n)$. 
The number of vertices and edges in the \LC instance produced by the reduction
would be $\poly(q^m)=2^{O( (\log n)^{b+5})}$.

We identify the variables of $\mathcal{I}$
with $S^m$ where $S \subseteq \F{q}$ is of size $h$. Any $\F{q}$ assignment $\sigma$
to the variables of $\mathcal{I}$ can be interpreted as an assignment 
$\sigma:S^m \to \F{q}$ and can be extended to a corresponding polynomial $g:\F{q}^m \to \F{q}$
of degree at most $(h-1)$ in each of the $m$ coordinates. Let $\mathcal{C}$ denote
the set of constraints of $\mathcal{I}$, each over $l:=5$
variables. Denote every constraint $C \in \mathcal{C}$ as
$C\left[ \left\{ x_i \right\}_{i=1}^{l} \right]$ where $\left\{ x_i \right\}_{i=1}^{l}$
is the set of points in $S^m$ corresponding to the $l$ variables of $\mathcal{I}$
on which the constraint is defined.

\begin{definition}
  A curve ${\omega}:\F{q} \to \F{q}^m$ of degree $r$ is a mapping
  ${\omega}(t):=(\omega_1(t),\cdots,\omega_m(t))$ where each $\omega_j$ is a degree $r$
  univariate polynomial in $t$.
\end{definition}
For the rest of this section, fix distinct values $t_0^*,\cdots,t_l^* \in \F{q}$.
A degree $l$ curve ${\omega}$ is said to correspond to a constraint 
$C\left[ \left\{ {x_i} \right\}_{i=1}^{l} \right]$ and an additional point ${x}$
if ${\omega}(t_0^*)={x}$, and for $i=1,\ldots,l$, ${\omega}(t_i^*)={x_i}$.
We now define the notion of a \emph{ruled surface}.
\begin{definition}
  A ruled surface $R=R[{\omega},{y}]$ where ${\omega}(t)$ is a curve and
  ${y} \in \F{q}^m$ is a direction, is a surface parametrized by two parameters
  $t,s \in \F{q}$, where
  \[
  R[{\omega},{y}](t,s)={\omega}(t)+s{y}\,.
  \]
\end{definition}
For a constraint $C$, a point ${x}$ and a direction ${y}$, let
$R[{\omega},{y}]$ be a ruled surface where ${\omega}$ is the curve
of degree $l$ corresponding to $C$ and ${x}$. Let $\mathcal{R}_{C}$
be the class of all such ruled surface corresponding to constraint $C \in \mathcal{C}$,
and let $\mathcal{R}:=\cup_{C \in \mathcal{C}} \mathcal{R}_{C}$.
Suppose the assignment $g:\F{q}^m \to \F{q}$ is a global polynomial of degree $d$.
The restriction of $g$ to a ruled surface $R \in \mathcal{R}$ is a bivariate polynomial
--- in $t$ and $s$ --- of degree at most $d^*:=ld=5d$ in variable $t$ and at most $d$
in variable $s$. The total number of coefficients of such a bivariate polynomial
is at most $d^* d=O(d^2)$. The \LC instance $\mathcal{L}$ is constructed as follows:
\begin{description}
  \item[Left vertex set $U$] This consists of all points in $\F{q}^m$. The label set is the set
    of values $\F{q}$ that can be assigned to the points.
  \item[Right vertex set $V$] The set of right vertices is $\mathcal{R}$, the class of all
    ruled surfaces over all constraints $C \in \mathcal{C}$.
    The label set of a ruled surface $R \in \mathcal{R}$ is the set of all bivariate polynomials
    in $t$ and $s$ of degree at most $d^*$ in $t$ and at most $d$ in $s$. Such a polynomial
    is represented by a vector of its coefficients. As mentioned before, the dimension of
    this vector is $O(d^2)$. For a ruled surface $R$ corresponding to a constraint $C$,
    there is a quadratic equation on these coefficients,
    and a label $g$ is valid --- $g \in \Gamma(R)$ --- iff the values of $g$
    at points $\left\{ (t^*_i, s=0) \right\}_{i=1}^{l}$ satisfy $C$.
  \item[Edges] For every ruled surface $R \in \mathcal{R}$ and every point $x \in R$,
    there is an edge between $x$ and $R$. The edge is satisfied by a labeling $g$
    to the surface $R$ and a label $p$ to the point $x$ if $g(x)=p$.
    Note that the computation of $g(x)$ is linear in the coefficients of $g$.
\end{description}

To analyze the above construction, we need the following result from \cite{ksspeccc}.
\begin{obsv}\label{obs:curveuniform}
  Given an equation $C$, pick a uniformly random point $x \in \F{q}^m$,
  and let ${\omega}$ be the curve corresponding to $C$ and ${x}$.
  Then, for any $t \in \F{q} \setminus \left\{ t^*_1,\cdots,t^*_l \right\}$,
  the point ${\omega}(t)$ is uniformly distributed in $\F{q}^m$.
\end{obsv}

\begin{theorem}\label{thm:lcfq}
  Let $k=(\log n)^b$, $q=2^{(\log n)^{b+4}}$ be as in Theorem \ref{thm:octsa-ext},
  and $\delta=2^{-(\log n)^{b+3}}$. 
  There is a reduction that takes as input a \ksat{3} instance of size $n$, 
  and outputs a \LC instance with parameters $h,m,d,l,d^*$ as described above.
  The instance $\mathcal{L}$ has the following properties:
  \begin{enumerate}
    \item The label set for $u \in U$ is $\F{q}$, and the label set for $v \in V$
      is the set of vectors of length $O(d^2)=O( (\log n)^5)$ over $\F{q}$.
      For each edge $e=(u,v)$, the projection $\pi^e$
      map coefficient vectors of surface polynomials
      to their value at points on the surface, and are homogeneous and $\F{q}$-linear.
      The coefficient vector is also supposed to satisfy an equation $C$ given by
      a quadratic equation over $\F{q}$. 
      
      The size of $\mathcal{L}$ is $\poly(q^m)=2^{O( (\log n)^{b+5})}$.
    \item If the \ksat{3} instance is satisfiable, then there is a labeling to 
      the vertices of $\mathcal{L}$ that satisfies all the edges, and that the label
      for every ruled surface $R$ satisfies the associated quadratic equation.
    \item If the \ksat{3} instance is unsatisfiable, then the following cannot hold
      simultaneously:
      \begin{itemize}
        \item For every left vertex $x \in \F{q}^m$, there are $k$ labels 
          $p_1^{x},\cdots,p_k^{x} \in \F{q}$.
        \item For every right vertex $R \in \mathcal{R}$, there are $k$ labels
          (polynomials given as coefficient vectors) ${g}^R_1,\cdots,{g}^R_k$,
          such that the associated equation is satisfied in superposition by
          those $k$ labels.
        \item For $\left( 1-\frac{1}{1000k} \right)$ fraction of the edges of $\mathcal{L}$,
          between a point $x$ and a ruled surface $R$, we have
          \[
          {g}^R_j(x)=p^{x}_j, \quad \forall j \in \left\{ 1,\ldots,k \right\}\,.
          \]
      \end{itemize}
      \item $\delta$-Smoothness: For any surface $R$, let ${g}$ be a non-zero label.
        Then
        \[
        \Pr_{x \in R}[{g}({x})=0] \le \delta\,.
        \]
  \end{enumerate}
\end{theorem}

\begin{proof}
  The reduction is described as above. Let $\mathcal{I}$ be a system of quadratic
  equations over $\F{q}$ given by Theorem \ref{thm:octsa-ext} and let $\mathcal{C}$
  be the set of equations.

  The $\delta$-smoothness property follows from Schwarz-Zippel lemma: 
  any non-zero label of surface $R$ gives a non-zero polynomial $g$, and its evaluation
  at a random point on the surface is zero with probability at most $d/q \le \delta$.
  
  In the completeness case, there is an assignment to each variable in $\mathcal{I}$.
  Therefore, there is a degree $d$ polynomial
  $g:\F{q}^m \to \F{q}$ that gives this assignment to the corresponding points in
  $S^m$. The left vertices are labeled using assignments given by $g$, and each
  ruled surface $R$ in the right vertex set are labeled by the polynomial given by
  the restriction of $f$ to $R$. This assignment satisfies the mapping on the edges,
  as well as the quadratic equation associated with $R$.

  For the soundness case, suppose for contradiction that 
  no $k$ assignments satisfy more than a $15/16$ fraction of equations in $\mathcal{I}$,
  but there are $k$ labelings
  for the vertices of $\mathcal{L}$, such that all associated equations of the right
  vertices are satisfied in superposition, and the mappings on the edges are
  satisfied for a $\left( 1-\frac{1}{1000 k} \right)$ fraction of the edges.
  By an averaging argument, we have that for a $\left( 1-\frac{1}{20} \right)$ fraction
  of the equations $C \in \mathcal{C}$, we have
  \[
  \Pr_{\substack{R \in \mathcal{R}_{C} \\ x \in R}}
  \left[ \bigwedge_{j=1}^{k}\left( {g}^R_j(x)=p^{x}_j \right) \right]
  \ge 1-\frac{1}{50 k}\,.
  \]
  We say that the equations $C$ that satisfy the condition above are \emph{good}.
  Fix one \emph{good} equation $C$. We say that a line $l(s)$ is contained in a ruled
  surface $R(t,s)$ if it is obtained by fixing a value of $t$ in $R(t,s)={\omega}(t)+s{y}$.
  Since choosing a random point on a ruled surface is equivalent to choosing a random line
  $l$ contained in $R$ then choosing a random point on $l$, the above probability can
  be rewritten as
  \[
  \Pr_{\substack{R \in \mathcal{R}_{C} \\ l \in R, x \in l}}
  \left[ \bigwedge_{j=1}^{k}\left( {g}^R_j(x)=p^{{x}}_j \right) \right]
  \ge 1-\frac{1}{50 k}\,.
  \]
  From Observation \ref{obs:curveuniform}, we can see that the above probability is essentially
  equal to the probability obtained by first picking a random line $l$ and then a random
  $R \in \mathcal{R}_C$ containing the line. Let $\mathcal{R}^l_C$ be the set of
  ruled surfaces of $R \in \mathcal{R}_C$ that contain line $l$. Thus we have
  \begin{equation}
    \Pr_{\substack{l, x \in l \\ R \in \mathcal{R}^l_C}}
    \left[ \bigwedge_{j=1}^{k}\left( {g}^R_j({x})=p^{{x}}_j \right) \right]
    \ge 1-\frac{2}{50 k}\,.
    \label{eq:curve}
  \end{equation}

  We now argue that for each $j=1,\ldots,k$, there exists a unique polynomial
  $P_j$ of total degree at most $d$, such that
  \begin{equation}
    \Pr_{\substack{R \in \mathcal{R}_C \\ {x} \in R}}
    \left[ \bigwedge_{j=1}^{k} \left( {g}^R_j({x})=P_j({x}) \right) \right]
    \ge \frac{9}{10}\,.
    \label{eq:curvedecoding}
  \end{equation}

  Define $f_1,\cdots,f_k:\F{q}^m \to \F{q}$ as $f_j(x)=p^{{x}}_j$ for all
  ${x} \in \F{q}^m$ and $j \in \left\{ 1,\ldots,k \right\}$.
  For each line $l$, construct $H(l)=(h_1(l),\cdots,h_k(l))$ as follows:
  first choose a random $R \in \mathcal{R}^l_C$, and let $h_j(l)$ be the univariate
  restriction of the bivariate polynomial ${g}^R_j$ to the line $l$.
  Let $E$ be the following event (over the choice of $l$, ${x} \in l$, and $H$):
  \[
  E \equiv \bigwedge_{j=1}^{k}\left( f_j({x}) = h_j(l)({x}) \right)\,.
  \]
  Equation (\ref{eq:curve}) can be re-stated as $\Pr_{l,{x} \in l,H}[E] \ge 1-\frac{2}{50k}$.
  The choice of $H$ is independent of $l$ and ${x}$, therefore there is a deterministic
  setting of $H$ for which $\Pr_{l,{x} \in l}[E] \ge 1-\frac{2}{50k}$.
  Applying Theorem \ref{thm:lowdegreelowerror}, we obtain polynomials
  $P_1,\cdots,P_k:\F{q}^m \to \F{q}$, each of total degree at most $d$, such that for each 
  $j \in \left\{ 1,\ldots,k \right\}$, we have
  \[
  \Pr_{{x}}[f_j({x}) \ne P_j({x})] = \Pr_{{x}}[p^{{x}}_j \ne P_j({x})]
  \le \frac{4}{50k}\,.
  \]
  Note that the choice of polynomials are globally unique 
  --- we get the same set of polynomials regardless of which good equation we fix ---
  due to Schwarz-Zippel lemma
  and that $d/q \ll O(1/k)$.
  
  Observe that a random point on a random ruled surface is essentially distributed uniformly
  randomly in $\F{q}^m$. Therefore, using Equation (\ref{eq:curve}) and the above,
  along with a union bound over the $k$ polynomials, we get Equation (\ref{eq:curvedecoding}).

  From Equation (\ref{eq:curvedecoding}), we have that there is one ruled surface
  $R \in \mathcal{R}_C$ such that
  \[
  \Pr_{x \in R}
  \left[ \bigwedge_{j=1}^{k} \left( {g}^R_j({x})=P_j({x}) \right) \right]
  \ge \frac{9}{10}\,.
  \]
  Using Schwarz-Zippel again, we have that for each $j \in \left\{ 1,\ldots,k \right\}$,
  the polynomial ${g}^R_j$ must be the restriction of $P_j$ to the surface $R$,
  and thus $P_j$ is consistent with ${g}^R_j$ at the points
  $(t^*_1,0),\cdots,(t^*_l,0)$. This means that the values given by $P_1,\cdots,P_k$
  satisfy the equation $C$ in superposition. 
  This holds for each good equation $C$, and therefore the assignment given by
  $P_1,\cdots,P_k$ satisfy in superposition a $19/20$ fraction of all quadratic equations
  in $\mathcal{I}$, contradicting the assumption about $\mathcal{I}$
  that no $k$ assignments satisfy in superposition more than a $15/16$ fraction of the
  equations.
\end{proof}

As in \cite{kssp}, we abstract out the above \LC and get the following statement
for \LC with $\F{2}$ vector labels. Moreover, 
we boost the soundness by taking $O(k^4)$ rounds of parallel
repetition and apply Theorem \ref{thm:parrep}.
It is important that we apply the parallel repetition theorem from 
\cite{parreprao}, because the label size here is large. 
Unlike earlier versions of parallel repetition theorem, 
the statement from \cite{parreprao} is independent of the label size and therefore
is efficient for our purpose.
Note that the smoothness property is preserved by parallel repetition.
\begin{theorem}\label{thm:lcvec}
  Let $k=(\log n)^b$ be as in Theorem \ref{thm:octsa-ext},
  and let $\delta=2^{-(\log n)^{b+3}}$. 
  There is a reduction that takes as input a \ksat{3} instance of size $n$, 
  and outputs a \LC instance $\mathcal{L}$ with the following properties:
  \begin{enumerate}
    \item The label set for $U$ is $\F{2}^{m_l}$,
      and the label set for $V$ is $\F{2}^{m_r}$,
      where $m_l,m_r=(\log n)^{5b+O(1)}$.

      For each vertex in $V$, there is a set of quadratic equations over $\F{2}$
      associated with it, and the coefficient vector is supposed to satisfy all
      the associated equations.

      The size of the vertex set of 
      $\mathcal{L}$ is $2^{O( (\log n)^{5b+O(1)})}$.

      For each edge $e$, the mapping $\pi^e$ is homogeneous and $\F{2}$-linear
      in the entries of the vectors, that is, there is a matrix $A_e \in \F{2}^{m_l \times m_r}$,
      such that for $x \in \F{2}^{m_l}$ and $y \in \F{2}^{m_r}$, 
      $\pi^e(y)=x$ iff $x=A_e y$.
    \item If the \ksat{3} instance is satisfiable, then there is a labeling to 
      the vertices of $\mathcal{L}$ that satisfies all the edges, and that the label
      to the right vertices satisfy all associated quadratic equations.
    \item If the \ksat{3} instance is unsatisfiable, then the following cannot hold
      simultaneously:
      \begin{itemize}
        \item For every left vertex $u \in U$, there are $k$ labels 
          $x^u_1, \cdots, x^{u}_k$.
        \item For every right vertex $v \in V$, there are $k$ labels
          ${y}^{v}_1,\cdots,{y}^{v}_{k}$ 
          that satisfy in superposition all the associated equations.
        \item For $2^{-(\log n)^{2b+O(1)}}$ fraction of the edges $e=(u,v)$ of $\mathcal{L}$,
          we have that for all $j \in \left\{ 1,\cdots,k \right\}$,
          $\pi^e({y}^{v}_j)={x}^{u}_j$.
      \end{itemize}
      \item $\delta$-Smoothness: For any $v \in V$ and non-zero label ${y}^v$,
        over the choice of a random edge incident on $v$, we have
        \[
        \Pr_{u \sim v}[\pi^{(u,v)}({y}^v)=0] \le \delta\,.
        \]
  \end{enumerate}
\end{theorem}
\begin{proof}[Proof Sketch]
  The first step is to translate the \LC from Theorem \ref{thm:lcfq} with $\F{q}$ vector labels
  to one with $\F{2}$ vector labels. This is done by choosing an arbitrary $\F{2}^{\log q}$ vector representation 
  for elements of $\F{q}$. This gives a $\F{2}$-vector-labeled \LC instance with the same
  guarantees as in Theorem \ref{thm:lcfq}, but with label sets $\F{2}^{m_l}$ and $\F{2}^{m_r}$ for $U$ and $V$,
  respectively, and $m_l,m_r = (\log n)^{b+O(1)}$.

  The second part of the reduction is to apply Theorem \ref{thm:parrep}, where we take the number of rounds $n$
  in Theorem \ref{thm:parrep} to be such that the repeated game has soundness $2^{-(\log n)^{2b+O(1)}}$.
  This requires $n=O(k^4)$. This gives the final \LC instance, and the
  procedure increases $m_l$ and $m_r$ by a multiplicative factor of $n$.

  To see that smoothness is preserved by parallel repetition, suppose we have a vertex
  $v:=(v_1,\cdots,v_t)$ and non-zero label
  $y:=(y_1,\cdots,y_t)$ for the $t$-round repeated game. Then there exists $s \in \left\{ 1, \ldots, t \right\}$
  such that $y_s \ne 0$. Also, for a uniformly random neighbor $u:=(u_1,\cdots,u_t)$ of $v$, the vertex
  $u_s$ is also a uniformly random neighbor of $v_s$. 
  Note that a necessary condition for $\pi^{(u,v)}(y)=0$ is that $\pi^{(u_s,v_s)}(y_s)=0$, and this probability
  is upper-bounded by the smoothness property of Theorem \ref{thm:lcfq}.
\end{proof}

We now construct the final \LC with matrix labels as follows.
\begin{theorem}\label{thm:lcmatrix}
  Let $k=(\log n)^b$ be as in Theorem \ref{thm:octsa-ext}.
  There is a reduction that takes as input a \ksat{3} instance of size $n$, 
  and outputs a \LC instance with the following properties:
  \begin{enumerate}
    \item The label set for the left vertex set $U$ is $\F{2}^{(m_l+1) \times (m_l+1)}$,
      and the label set for the right vertex set $V$ is $\F{2}^{(m_r+1) \times (m_r+1)}$,
      where $m_l,m_r=(\log n)^{5b+O(1)}$.

      For each right vertex, there is a set of homogeneous linear $\F{2}$
      equations involving entries of the labeling of $v$. The set of valid labelings
      $\Gamma(v)$ consists of matrices that satisfy all the associated linear equations.
      
      The size of the vertex set of $\mathcal{L}$ is $2^{O( (\log n)^{5b+O(1)})}$.

      For each edge $e=(u,v)$, there is a matrix $A_e$, such that labeling $M_u$ and $M_v$
      satisfies $\pi^e$ iff $M_u=A M_v A^T$. Note that the constraint is linear in the entries
      of $M_u$ and $M_v$.
    \item If the \ksat{3} instance is satisfiable, then there is a labeling to 
      the vertices of $\mathcal{L}$ that satisfies all the edges, and that 
      for each right vertex $v \in V$, its label $M_v$ is symmetric, $\rank(M_v)=1$,
      the $(1,1)$ entry of $M_v$ is $1$, and $M_v \in \Gamma(v)$.
    \item If the \ksat{3} instance is unsatisfiable, then for any labeling $\sigma$
      for the vertices in $U$ and $V$, the following cannot hold
      simultaneously:
      \begin{itemize}
        \item For each $v \in V$, the matrix $\sigma(v)$ is pseudo-quadratic, has
          $\rank(\sigma(v)) \le (\log n)^b/2$, and is valid $\sigma(v) \in \Gamma(v)$.
        \item For at least $2^{-(\log n)^{b}}$ fraction of the edges $e=(u,v)$ of $\mathcal{L}$,
          we have that $\pi^e(\sigma(v))=\sigma(u)$.
      \end{itemize}
      \item Smoothness: For any $v \in V$ and $M_v$ such that $D_1(M_v) \ne 0$,
        then over the choice of a random edge incident on $v$, we have
        \[
        \Pr_{u \sim v}[D_1(\pi^{(u,v)}(M_v)=0)] \le 2^{-(\log n)^{b+1}}\,.
        \]
  \end{enumerate}
\end{theorem}
\begin{proof}
  We start with the \LC instance from the previous theorem.

  The underlying bipartite graph of the new instance is exactly the same. 
  The parameters $m_r$ and $m_l$ are the same as before. The labels for $u \in U$
  in the new instance are now matrices from $\F{2}^{(m_l+1) \times (m_l+1)}$, and
  the labels for $v \in V$ are from $\F{2}^{(m_r+1) \times (m_r+1)}$. 
  The constraints for labelings for vertices in $v \in V$ are the following:
  \begin{enumerate}
    \item The matrix label $M$ is symmetric, and for $i=2,\cdots,m_r+1$,
      we have $M_{i,i}=M_{1,i}=M_{i,1}$. These are all homogeneous linear constraints.
      Note that if in addition we have $M_{1,1}=1$, then we get 
      that $M$ is pseudo-quadratic. Here, however, we do not include the latter constraint
      as it is not homogeneous. In fact, this will be handled by the inner verifier.
    \item For each quadratic constraint in the previous instance, we include the
      linearized version of it in the new instance. That is, term $x_i x_j$ is replaced
      by entry $(i+1,j+1)$ of the matrix, term $x_i$ is replaced by entry $(1,i+1)$,
      and constant $1$ is replaced by entry $(1,1)$.
  \end{enumerate}
  For an edge $e$, let $A'_e$ be the matrix in the \LC instance from Theorem \ref{thm:lcvec}, 
  then we define the matrix for the new instance to be $A_e=\left(
  \begin{array}{c c}
    1 & 0 \\
    0 & A'_e
  \end{array}
  \right)$. Let $M_v=\left(
  \begin{array}{c c}
    a & \alpha \\
    \beta & D
  \end{array}
  \right)$ be a label. Then it is mapped to
  \[
    \pi^e(M_v)=A_e M_v A_e^T = \left(
    \begin{array}{c c}
      a & \alpha A'^T_e \\
      A'_e \beta & A'_e D A'^T_e
    \end{array}
    \right)\,.
  \]

  In the completeness case, a vector label $\alpha$ in the previous theorem
  is transformed into the matrix label $(1 \ \alpha)(1 \ \alpha)^T$.

  For the soundness case, suppose that there
  are pseudo-quadratic matrices $M_u$ and $M_v$ for each $u \in U$ and $v \in V$,
  such that $M_v$ satisfies homogeneous linear constraints associated with $v$, 
  $\rank(M_v) \le k$, and that for $2^{-(\log n)^b}$ fraction of the edges $e$, 
  $\pi^e(M_v)=M_u$.

  For vertex $v \in V$, by Lemma \ref{lem:rankoddcov},
  there exists odd integer $l<3/2 \cdot (\log n)^b/2<(\log n)^b$ vectors 
  $y'_1,\cdots,y'_l \in \F{2}^{m_r+1}$, where 
  $y'_{i,1}=1$ for $i \in [l]$, such that $M_v=\sum_{i=1}^{l} {y}'_i \otimes {y}'_i$,
  and the assignments $y_1:=D_1({y}_1),\cdots,y_l:=D_1({y}_l)$ satisfy in superposition the quadratic
  constraints of the \LC instance from Theorem \ref{thm:lcvec}. 
  By padding $0$ assignments, we can make sure that we have
  exactly $(\log n)^b$ assignments $y_1,\cdots,y_{(\log n)^b}$ that satisfy in superposition
  the quadratic constraints of the \LC instance from Theorem \ref{thm:lcvec}, 
  and for all $j \in \left[ (\log n)^b \right]$,
  $y_j$ is in the column space of $D_1(M_v)$.

  For the decoding of vertices in $U$, we use the following lemma from \cite{ksspeccc}
  (Lemma 7.3), adapted to our choice of parameters. The proof is identical
  and we include it here for sake of completeness.
  \begin{lemma}
    Fix any $v \in V$, a rank parameter $l \le (\log n)^b$, 
    and a matrix $M \in \F{2}^{(m_r+1) \times (m_r+1)}$ such that $\rank(D_1(M))=l$.
    Then over the choice of a random neighbor
    $u$ of $v$, we have $\rank(D_1(\pi^{(u,v)}(M)))=l$ except with probability $2^{-(\log n)^{b+1}}$.
  \end{lemma}
  \begin{proof}
    Using Lemma 2.1 of \cite{ksspeccc}, $D_1(M)$ can be decomposed into the canonical form
    \[
    D_1(M)=\sum_{j=1}^{s} z_j \otimes z_j + \sum_{j=1}^{t}
    \left( z_{s+2j-1} \otimes z_{s+2j} + z_{s+2j} \otimes z_{s+2j-1} \right)\,,
    \]
    where $l=s+2t$ is the rank of $D_1(M)$, and $z_1,\cdots,z_l$ are linearly independent.
    For $j=1,\cdots,l$, define $z'_j:=A'_{(u,v)} z_j$. Then $D_1(\pi^{(u,v)}(M))$ can be
    written as
    \[
    D_1(\pi^{(u,v)}(M))=\sum_{j=1}^{s} z'_j \otimes z'_j + \sum_{j=1}^{t}
    \left( z'_{s+2j-1} \otimes z'_{s+2j} + z'_{s+2j} \otimes z'_{s+2j-1} \right)\,,
    \]
    Consider a non-zero linear combination $z$ of the vectors $\left\{ z_j \right\}_{j=1}^{l}$.
    The vector $A'_{(u,v)} z$ is the corresponding linear combination of
    the vectors $\left\{ z'_j \right\}_{j=1}^{l}$, and is non-zero with probability 
    $2^{-(\log n)^{b+3}}$ by the smoothness guarantee from Theorem \ref{thm:lcvec}.
    Taking a union bound over all $2^l-1$ non-zero linear combination, we conclude
    that except with probability at most $2^{-(\log n)^{b+1}}$, the vectors
    $\left\{ z'_j \right\}_{j=1}^{l}$ are also linearly independent and spans
    the column space of $D_1(\pi^{(u,v)}(M))$.
  \end{proof}

  The smoothness property of the new \LC instance follows easily from the above lemma.
  
  For each $u$, we choose $(\log n)^b$ uniformly random vectors ${x}_1,\cdots,{x}_{(\log n)^b}$
  from the column space of $D_1(M_u)$.  %  , and let ${x}_j=(1 \ {x}'_j)$ for $j \in \left\{ 1,\cdots,(\log n)^b \right\}$.
  Now we analyze the expected value of this assignment.

  Note that here the soundness parameter $2^{-(\log n)^b} \gg 2^{-(\log n)^{b+1}}$. Therefore, by the above lemma,
  for $2^{-(\log n)^{b+o(1)}}$ fraction of the edges $e=(u,v)$, we have 
  $\rank(D_1(M_v))=\rank(D_1(\pi^{e}(M_v)))=\rank(D_1(A_{e} M_v A_{e}^T))=\rank(A'_{e} D_1(M_v) A'^T_{e})$.

  Fix such an edge. 
  Since $\rank(D_1(M_v))=\rank(A'_e D_1(M_v) A'^T_e) \le \rank(A'_e D_1(M_v)) \le \rank(D_1(M_v))$,
  this means that $\rank(A'_e D_1(M_v) A'^T_e) = \rank(A'_e D_1(M_v)) = \rank(D_1(M_v))$.
  For any $y$ that is in the column space of $D_1(M_v)$, $A'_e y$ is in the column space of $A'_e D_1(M_v)$,
  and since $\rank(A'_e D_1(M_v))=\rank(A'_e D_1(M_v) A'^T_e)$, we conclude that $A'_e y$ is also in the column space
  of $A'_e D_1(M_v) A'^T_e=D_1(M_u)$.
  Thus for edge $e$, with probability at least $2^{-(\log n)^{2b}}$,
  we get $\pi(y_j)=A'_e y_j=x_j$ for all $j \in \left\{ 1,\cdots,l \right\}$.

  Overall, this labeling satisfies 
  $2^{-(\log n)^{b+o(1)}} 2^{-(\log n)^{2b}} = 2^{-(\log n)^{2b+O(1)}}$ fraction
  of the edges in the old instance.
\end{proof}

\section{Hypergraph Coloring Hardness}
We now compose the \LC from Theorem \ref{thm:lcmatrix}
with a \QDCD inner-verifier to get inapproximability result
for hypergraph coloring.

\begin{theorem}
  There is a reduction that takes as input a \ksat{3} instance
  of size $n$, outputs a $8$-uniform hypergraph $H$ with the following properties:
  \begin{itemize}
    \item The size of the hypergraph $H$ and the running time of the
      reduction are both upper-bounded by $\exp( (\log n)^{(10+o(1))b})$.
    \item If the \ksat{3} instance is satisfiable, then $H$ is $2$-colorable.
    \item If the \ksat{3} instance is unsatisfiable, then $H$ does not have
      independent set of fractional size larger than $2^{-O( (\log n)^b)}$.
  \end{itemize}
  In other words, it is quasi-\np-hard to color a $2$-colorable
  $8$-uniform hypergraph of size $N$ with less than $2^{(\log N)^{1/10-o(1)}}$
  colors.
\end{theorem}

The following proof is based on a note by Girish Varma \cite{varmasp}.

Given the \LC instance from Theorem \ref{thm:lcmatrix}, we expect for each
vertex $v \in V$ a function $f_v:\F{2}^{(m_r+1) \times (m_r+1)} \to \F{2}$.
The expected encoding for matrix label $\sigma(v)=a_v \otimes a_v$ is
$f_v(A)=\langle a_v \otimes a_v, A \rangle=a_v^T A a_v$.
Let $\mathcal{H}_v \subseteq \F{2}^{(m_r+1) \times (m_r+1)}$ be the 
dual of the subspace of the set of pseudo-quadratic matrices
that satisfies the linear constraints associated with $v$.
The function $f_v$ is folded over $\F{2}^{(m_r+1) \times (m_r+1)} / \mathcal{H}_v$.

Consider the following Boolean $8$-uniform test:
\begin{itemize}
  \item Choose $u \in U$ uniformly at random, and $v,w \in V$ uniformly and
    independently at random from the neighbors of $u$.
    Let $\pi,\sigma:\F{2}^{(m_r+1) \times (m_r+1)} \to \F{2}^{(m_l+1) \times (m_l+1)}$
    be the projections corresponding to the edges
    $(u,v)$ and $(u,w)$ respectively, and let $S_{\pi}$ and $S_{\sigma}$ be the
    index set associated with them.
  \item Uniformly and independently sample $X_1,X_2,Y_1,Y_2 \in \F{2}^{(m_r+1) \times (m_r+1)}$,
    $F \in \F{2}^{(m_l+1) \times (m_l+1)}$,
    and $x,y,z,x',y',z' \in \F{2}^{m_r+1}$. 
    Let $e \in \F{2}^{m_r+1}$ be the vector with only the $1$-st entry $1$ and the rest
    $0$.
  \item Accept if and only if the following $8$ values are not all equal:

    \begin{tabular}{l l l}
      $f_v(X_1)$ & $f_v(X_3)$ & where $X_3:=X_1+x \otimes y+F \circ \pi$ \\
      $f_v(X_2)$ & $f_v(X_4)$ & where $X_4:=X_2 + (x+e) \otimes z+F \circ \pi$ \\
      $f_w(Y_1)$ & $f_w(Y_3)$ & where $Y_3:=Y_1+x' \otimes y'+F \circ \sigma + e \otimes e$ \\
      $f_w(Y_2)$ & $f_w(Y_4)$ & where $Y_4:=Y_2+(x'+e) \otimes z'+F \circ \sigma+e \otimes e$
    \end{tabular}
\end{itemize}

We denote by $\mathcal{T}$ the test distribution.

The vertex set of the hypergraph has size 
\[\exp((\log n)^{(5+o(1))b})) \cdot 2^{(\log n)^{2(5+o(1))b}}=\exp((\log n)^{(10+o(1)) b}))=:N\,.\]

\subsection{Completeness}
Let $y_v \otimes y_v$ for $v \in V$ and $x_u \otimes x_u$ for $u \in U$ be a perfect labeling
for the Label Cover instance, with $y_{v,1}=x_{u,1}=1$ and for each edge $e=\left\{ u,v \right\} \in E$,
we have $(y_v)|_{S_e}=x_u$. Consider the $2$-coloring
where for each $v \in V$, $f_v(X)=y_v^T X y_v=\langle X, y_v \otimes y_v \rangle$.
Such a function is constant over cosets of $\mathcal{H}_v$.
Let
$x_1:=\langle X_1,y_v \otimes y_v\rangle$, $x_2:=\langle X_2,y_v \otimes y_v\rangle$,
$y_1:=\langle Y_1,y_w \otimes y_w\rangle$, $y_2:=\langle Y_2,y_w \otimes y_w\rangle$, and
$f:=\langle F,x_u \otimes x_u\rangle$. 
Note that 
$\langle F,x_u \otimes x_u\rangle=\langle F,\pi_{u,v}(y_v \otimes y_v)\rangle=\langle F \circ \pi_{uv},y_v \otimes y_v\rangle$.
Also, $\langle e \otimes e,y_v \otimes y_v\rangle=\langle e,y_v \rangle=1$. Therefore, the value
of the $8$ queries are 

\vspace{.1in}

\begin{tabular}{l c l}
  $x_1$ & & $x_1+\langle y_v,x\rangle\langle y_v,y\rangle+f$ \\
  $x_2$ & & $x_2+(\langle y_v,x\rangle+1)\langle y_v,z\rangle+f$ \\
  $y_1$ & & $y_1+\langle y_w,x'\rangle\langle y_w,y'\rangle+f+1$ \\
  $y_2$ & & $y_2+(\langle y_w,x'\rangle+1)\langle y_w,z'\rangle+f+1$ \\
\end{tabular}

\vspace{.1in}

We finish the proof of the completeness case by a case analysis.

If $\langle y_v,y \rangle=\langle y_w,y'\rangle=0$, then the sum of entries in the first and third row 
is $1$, which means that there are different values.
Similarly, we conclude that if $\langle y_v,z \rangle=\langle y_w,z' \rangle=0$, then using similar argument as above,
there are different values in the second and the fourth row.
The same applies to the case when $\langle y_v,x\rangle=\langle y_2,x'\rangle=1$, and
the case when $\langle y_v,x\rangle=\langle y_w,x'\rangle=0$.

Suppose now that $\langle y_v,x \rangle=1$ and all entries are equal. Then from the second row, we have
that $f=0$, and from the first row, we get $\langle y_v,y \rangle=0$. By the discussion above, we have
that $\langle y_w,y' \rangle=1$, and the third row gives us $\langle y_w,x'\rangle =1$, but then the
two entries on the last row are different.

Suppose otherwise that $\langle y_v,x\rangle =0$ and all entries are equal. Then from the first row, we have
$f=0$, and the second row implies $\langle y_v,z \rangle=0$. By the discussion above, we must
have $\langle y_w,z' \rangle=1$, and the last row gives $\langle y_w,x' \rangle=0$, leaving two
different entries in the third row.

Hence $f_v$ gives a valid $2$-coloring of $\mathcal{G}$.

\subsection{Soundness}
Let $\delta=2^{-(\log n)^b}$ be the soundness parameter from Theorem \ref{thm:lcmatrix}
and $k=(\log n)^b/2$ be the rank upper-bound from Theorem \ref{thm:lcmatrix}.

\begin{lemma}
  If there is an independent set in $\mathcal{G}$ of relative size $s$, then
  \[
  s^8 \le \delta+\frac{1}{2^{k/2+1}}\,.
  \]
\end{lemma}
\begin{proof}
  Consider any set $A \subseteq \mathcal{V}(\mathcal{G})$ of fractional size $s$.
  For every $v \in V$, let $f_v:\F{2}^{(m_r+1) \times (m_r+1)} \to [0,1]$ be the indicator
  function of $A$, extended such that it is constant over cosets of $\mathcal{H}_v$.
  The fractional size of $A$ is given by
  \begin{align*}
    \E_{\substack{v \sim V \\ X \sim \F{2}^{(m_r+1) \times (m_r+1)}}}
    \left[ f_v(X) \right] \ =&\  \E_{v \sim V}\left[ \widehat{f}_{v,0}\right]\,.
  \end{align*}

  The set $A$ is an independent set if and only if
  \begin{equation}
    \Theta:=\E_{u,v,w} \E_{X_i,Y_i \sim \mathcal{T}} \prod_{i=1}^{4} f_v(X_i)f_w(Y_i) = 0\,.
    \label{eq:hcolmain}
  \end{equation}
  Taking Fourier expansion and considering expectations over $X_1,X_2,Y_1,Y_2$, we get the following:
  \begin{align*}
    \Theta \ =&\ 
    \E_{u,v,w} \sum_{\alpha_1,\alpha_2,\beta_1,\beta_2 \in \F{2}^{(m_r+1) \times (m_r+1)}} \E_{F,x,x'}
    \Bigg[ \\
    &\  \widehat{f}_{v,\alpha_1}^2 \E_{y}[\chi_{\alpha_1}(x \otimes y)] \chi_{\alpha_1}(F \circ \pi) \\
    &\  \widehat{f}_{v,\alpha_2}^2 \E_{z}[\chi_{\alpha_2}( (x+e) \otimes z)] \chi_{\alpha_2}(F \circ \pi) \\
    &\  \widehat{f}_{w,\beta_1}^2 \E_{y'}[\chi_{\beta_1}(x' \otimes y')]  
    \chi_{\beta_1}(F \circ \sigma)\chi_{\beta_1}(e \otimes e) \\
    &\  \widehat{f}_{w,\beta_2}^2 \E_{z'}[\chi_{\beta_2}( (x'+e) \otimes z')] 
    \chi_{\beta_2}(F \circ \sigma)\chi_{\beta_2}(e \otimes e)
    \Bigg]\,.
  \end{align*}
  Denote the term inside $\E_{F,x,x'}[\cdot]$ as $Term_{u,v,w}(\alpha_1,\alpha_2,\beta_1,\beta_2)$.

  For the characters involving $F$, we have
  \begin{align*}
    &\  \E_{F}\left[ \chi_{\alpha_1}(F \circ \pi)\chi_{\alpha_2}(F \circ \pi)\chi_{\beta_1}(F \circ \sigma)
    \chi_{\beta_2}(F \circ \sigma) \right] \\
    =&\  
    \E_{F}\left[ (-1)^{\langle\pi(\alpha_1+\alpha_2),F\rangle+\langle\sigma(\beta_1+\beta_2),F\rangle} \right]\,,
  \end{align*}
  and since $F \in \F{2}^{(m_l+1) \times (m_l+1)}$ is chosen uniformly at random,
  the above is $0$ unless $\pi(\alpha_1+\alpha_2)=\sigma(\beta_1+\beta_2)$.

  Let $\nu(\alpha):=\langle\alpha,e \otimes e\rangle$. Taking expectations over $x,y,z,x',y',z'$,
  we have that when $\pi(\alpha_1+\alpha_2) \ne \sigma(\beta_1+\beta_2)$, 
  $Term_{u,v,w}(\alpha_1,\alpha_2,\beta_1,\beta_2)=0$, and otherwise
  \begin{align*}
    &\  Term_{u,v,w}(\alpha_1,\alpha_2,\beta_1,\beta_2) \\
    =&\  (-1)^{\nu(\beta_1+\beta_2)} \widehat{f}_{v,\alpha_1}^2\widehat{f}_{v,\alpha_2}^2\widehat{f}_{w,\beta_1}^2
    \widehat{f}_{w,\beta_2}^2 \\ &\  \Pr_{x}\left[ \alpha_1 x= 0 \land \alpha_2 x=\alpha_2 e \right]
    \Pr_{x'}\left[ \beta_1 x=0 \land \beta_2 x' = \beta_2 e \right]\,.
  \end{align*}
  The terms that are potentially non-zero can now be partitioned into three parts:
  \begin{align*}
    \Theta_0 \ =&\ 
    \E_{u,v,w} \sum_{\substack{\rank(\alpha_1+\alpha_2),\rank(\beta_1+\beta_2) \le k \\ 
    \pi(\alpha_1+\alpha_2)=\sigma(\beta_1+\beta_2) \\ \nu(\beta_1+\beta_2)=0}}
    Term_{u,v,w}(\alpha_1,\alpha_2,\beta_1,\beta_2) \\
    \Theta_1 \ =&\ 
    \E_{u,v,w} \sum_{\substack{\rank(\alpha_1+\alpha_2),\rank(\beta_1+\beta_2) \le k \\ 
    \pi(\alpha_1+\alpha_2)=\sigma(\beta_1+\beta_2) \\ \nu(\beta_1+\beta_2)=1}}
    Term_{u,v,w}(\alpha_1,\alpha_2,\beta_1,\beta_2) \\
    \Theta_2 \ =&\ 
    \E_{u,v,w} \sum_{\substack{\max\{\rank(\alpha_1+\alpha_2),\rank(\beta_1+\beta_2)\} > k \\ 
    \pi(\alpha_1+\alpha_2)=\sigma(\beta_1+\beta_2)}}
    Term_{u,v,w}(\alpha_1,\alpha_2,\beta_1,\beta_2)\,.
  \end{align*}
  We first lower-bound $\Theta_0$. Note that all terms in $\Theta_0$ are positive.
  Consider the term corresponding to $\alpha_1=\alpha_2=\beta_1=\beta_2=0$. We have
  \[
  \E_{u,v,w}\widehat{f}_{v,0}^4 \widehat{f}_{w,0}^4=\E_{u}\left( \E_v \widehat{f}_{v,0}^4 \right)^2
  \ge \left( \E_{u,v} \widehat{f}_{v,0} \right)^8 \ge s^8\,.
  \]
  Therefore $\Theta_0 \ge s^8$.
  
  For $\Theta_1$, we have the following upper-bound
  \begin{equation}
    |\Theta_1| \le 
    \E_{u,v,w} \sum_{\substack{\rank(\alpha_1+\alpha_2),\rank(\beta_1+\beta_2) \le k \\ 
    \pi(\alpha_1+\alpha_2)=\sigma(\beta_1+\beta_2) \\ \nu(\beta_1+\beta_2)=1}}
    \widehat{f}_{v,\alpha_1}^2 \widehat{f}_{v,\alpha_2}^2 \widehat{f}_{w,\beta_1}^2
    \widehat{f}_{w,\beta_2}^2\,.
    \label{eq:hcoldecoding}
  \end{equation}
  Consider the following randomized labeling strategy for vertices in $u \in U$ and $v \in V$:
  for $v \in V$, pick $(\beta_1,\beta_2)$ with probability $\widehat{f}_{v,\beta_1}^2\widehat{f}_{v,\beta_2}^2$
  and set its label to $\beta_1+\beta_2$; for $u \in U$, pick a random neighbor $v$, and choose
  $(\alpha_1,\alpha_2)$ with probability $\widehat{f}_{v,\alpha_1}^2\widehat{f}_{v,\alpha_2}^2$
  and set its label to $\pi(\alpha_1+\alpha_2)$. Due to folding, we have that $\beta_1$ and $\beta_2$
  both satisfies the homogeneous linear constraints associated with $v$, and so does $\beta_1+\beta_2$.
  Therefore the right hand side of (\ref{eq:hcoldecoding}) gives the probability that a random
  edge of the Label Cover is satisfied by this labeling. Thus $|\Theta_1| \le \delta$.

  For $\Theta_2$, note that if $\rank(\alpha)>k$, then for any fixed $b$,
  $\Pr_x[\alpha x=b] \le 1/2^{k+1}$. Therefore, for any fixed choice of $u,v,w$,
  all terms in $\Theta_2$ have absolute value at most $1/2^{k/2+1}$. 
  Combined with Parseval's identity, we conclude that $|\Theta_2| \le 1/2^{k/2+1}$.
\end{proof}

We conclude that any independent set in $\mathcal{G}$ has fractional size at most
$2^{-\log^b n/32}$, and therefore the chromatic number of $\mathcal{G}$ is at least
$2^{\log^b n/32} = \exp( (\log N)^{1/(10-o(1))})$.

\section*{Acknowledgments}
I would like to thank Johan Håstad for numerous inspiring discussions.
I am also grateful to Rishi Saket who pointed out a mistake in an earlier version of 
this manuscript.

\bibliographystyle{plain}
\bibliography{oddcovernew-arxiv}

\end{document}